\documentclass[11pt,draftcls,onecolumn,romanappendices]{IEEEtran}

\newif\ifsubmission
\submissionfalse 

\usepackage[latin1]{inputenc}
\usepackage{amsmath}
\usepackage{amsfonts}
\usepackage{amssymb}
\usepackage{amstext}
\usepackage{amsthm}
\usepackage{subfigure}
\usepackage{graphicx}
\usepackage{url}
\usepackage{color}
\usepackage{xcolor}
\usepackage{bm}
\usepackage{dsfont}
\usepackage{pdfpages}
\usepackage{soul}
\usepackage[sort,nocompress]{cite} 

\newtheorem{theorem}{Theorem}
\newtheorem{lemma}{Lemma}

\newtheorem{corollary}{Corollary}
\newtheorem{definition}{Definition}
\newtheorem{property}{Property}

\newcommand{\rarrow}{\tiny \to}

\newcommand{\dfn}{\stackrel{\triangle}{=}}
\newcommand {\exe} {\stackrel{\cdot} {=}}
\newcommand {\lexe} {\stackrel{\cdot} {\le}}

\newcommand {\reals} {{\rm I\!R}}

\newcommand {\bx} {\mbox{\boldmath $x$}}
\newcommand {\by} {\mbox{\boldmath $y$}}

\newcommand {\bX} {\mbox{\boldmath $X$}}

\newcommand{\calA}{{\cal A}}

\newcommand{\calC}{{\cal C}}
\newcommand{\calD}{{\cal D}}

\newcommand{\calH}{{\cal H}}

\newcommand{\calS}{{\cal S}}
\newcommand{\calT}{{\cal T}}

\renewcommand{\SS}{\mathcal S}
\colorlet{mygreen}{green!60!gray}

\begin{document}

\title{Detection Games Under Fully Active
Adversaries}
%
\author{
Benedetta Tondi,  \IEEEmembership{Member, IEEE},  Neri Merhav,  \IEEEmembership{Fellow, IEEE}, Mauro Barni  \IEEEmembership{Fellow, IEEE}
\thanks{M. Barni and B. Tondi are with the
Department of Information Engineering and Mathematical Sciences,
University of Siena,  Siena, ITALY, e-mail: \{benedettatondi@gmail.com, barni@dii.unisi.it\};
N. Merhav is with the the Andrew
and Erna Viterbi Faculty of Electrical Engineering - Israel Institute of Technology Technion City, Haifa,
ISRAEL,
email: \{merhav@ee.technion.ac.il\}.}
}

\markboth{IEEE TRANSACTIONS ON INFORMATION THEORY, ~Vol.~X, No.~X, XXXXXXX~XXXX}%
{B. Tondi,  N.Merhav, M. Barni: Detection Games Under Fully Active
Adversaries}

\maketitle

\begin{abstract}
%
We study a binary hypothesis testing problem in which a defender must decide whether or not a 
test sequence has been drawn from a given memoryless source $P_0$ whereas,
an attacker strives to impede the correct detection.
With respect to previous works, the adversarial setup
addressed in this paper considers an attacker who is
active under both hypotheses, namely, a fully active attacker, as opposed
to a partially active attacker  who
is active under one hypothesis only.
In the fully active setup, the attacker distorts
sequences drawn both from $P_0$ and from an
alternative memoryless source $P_1$, up to a certain distortion level,
which is possibly different under the two hypotheses,
in order to maximize the confusion in
distinguishing between the two sources, i.e.,
to induce
both false positive and  false negative
errors
at the detector, also referred to
as the defender.
We model the defender-attacker interaction
as a game and study two versions of this game, the Neyman-Pearson game
and the Bayesian game.
Our main result
is  in the characterization of
an attack strategy that is asymptotically both dominant
(i.e., optimal no matter what
the defender's strategy is) and universal, i.e., independent of
 $P_0$ and $P_1$.
From the analysis of the equilibrium payoff, we
also derive the best achievable performance of the defender,
by relaxing the requirement on the exponential
decay rate of the false positive error probability
in the Neyman--Pearson setup and the  tradeoff
between the error exponents in the Bayesian setup.
Such analysis permits to
characterize the conditions for the distinguishability of the two sources
given the distortion levels.
\end{abstract}



\IEEEpeerreviewmaketitle

\section{Introduction}
\label{sec.intro}

There are many fields in signal processing and communications where the detection
problem should naturally be framed within an adversarial setting: multimedia forensics (MF) \cite{Boh12}, spam filtering   \cite{Lowd05},
biometric-based verification \cite{jain2005biometric},
one-bit watermarking \cite{TheGaussWatermGame},
and digital/analogue transmission under jammer attacks \cite{tan2002undermining},
just to name a few
(see \cite{barni2013coping} for other examples).

%
%
In particular,
the need for adversarial modeling has become evident
in security-related applications and  game theory is
often  harnessed as a useful tool
in  many research areas,
such as steganalysis \cite{ker2007batch},
watermarking \cite{TheGaussWatermGame},
intrusion detection systems \cite{wu2012anti}
and adversarial machine learning \cite{DDMSV04,goodfellow2014generative}.
%
%
In recent literature, game theory and information theory  have also been
combined to address the problem of adversarial detection,
especially in the field of digital watermarking,
see, for instance, \cite{Sullican_Moulin_98,TheGaussWatermGame,Sullican_Moulin_03, SMM04}.
In all these works, the problem of designing
watermarking codes that are robust to
intentional attacks, is studied as a game between the information hider
and the attacker.

An attempt to develop
a general theory for the binary hypothesis testing problem in
the presence of an adversary was made in \cite{BT13}.
Specifically, in \cite{BT13} the general problem of binary
decision under adversarial conditions has been
addressed
and formulated as a game between two players, the {\it defender}
and the {\it attacker}, which have conflicting
goals.
Given two discrete memoryless sources, $P_0$ and $P_1$,
the goal of the defender is to decide whether
a given
test sequence has been generated by $P_0$ (null hypothesis,
$\calH_0$) or $P_1$ (alternative hypothesis, $\calH_1$).
By adopting the Neyman-Pearson approach, the set of
strategies the defender can choose
from is the set of decision regions for $\calH_0$ ensuring that the false positive error probability is lower than a given threshold.
On the other hand, the ultimate goal of the
attacker in \cite{BT13} is to cause a false negative
decision,
so the attacker acts under $\calH_1$ only. In other words, the attacker
modifies
a sequence generated by $P_1$,
in attempt to move it into the acceptance region of
$\calH_0$.
The attacker is subjected to a distortion constraint,
which limits his freedom in doing so.
%
Such a struggle between the defender and the attacker is modeled in \cite{BT13} as a competitive zero-sum game
and the asymptotic equilibrium, that is, the equilibrium when the
length of the observed sequence tends to infinity, is derived under the
assumption that the defender bases his decision on the
analysis of first order statistics only.
In this respect, the analysis conducted in \cite{BT13} extends the one of
\cite{merhav_sabbag_08} to the adversarial scenario.
Some variants of this attack-detection game have also been
studied:
in \cite{BTtit}, the setting was extended to the case
where the sources are known to neither the defender nor the attacker, yet
training data from both sources is available to both parties:
within this framework, the case where  part of the
training data available to the defender is corrupted by the attacker
has also been studied (see \cite{BTtit18}).

There are many situations
in which it is reasonable to assume that the attacker is active
under both  hypotheses with the goal of causing both false positive and false negative detection
errors.
For instance, in applications of camera fingerprint
detection,
an adversary might be interested to remove the fingerprint from a given image
so that the generating camera would not
be identified and at the same time,
to implant the fingerprint from another camera \cite{Chen08,Gol11}.
Another example comes from watermarking, where an attacker can be interested in either removing or injecting the watermark from an image or a video, 
to redistribute the content with a fake copyright and no information
(erased information) about the true ownership \cite{moulin2001game}.
Attacks
under both hypotheses may also be present in applications of network intrusion detection \cite{corona2013adversarial}.
Network intrusion detection systems, in fact, can be subject to both evasion attacks  \cite{vigna2004testing}, in which an adversary tries to avoid detection by manipulating malicious traffic,
and overstimulation attacks \cite{patton2001achilles,mutz2005reverse}, in which the network is overstimulated by an adversary who sends synthetic traffic (matching the legitimate traffic)
in order to cause a denial of service.

With the above ideas in mind, in this paper,
we consider the game--theoretic
formulation of the defender-attacker interaction when the attacker
acts under both hypotheses.
We refer to this scenario as a detection game with
a {\it fully active attacker}. By contrast, when the attacker acts under
hypothesis $\calH_1$ only (as in \cite{BT13} and \cite{BTtit}), he is referred to as
a {\it partially active attacker}.
%
A distinction is made between the case where the underlying hypothesis
is known to the attacker and the case where it is not. A little thought,
however, immediately indicates that the latter is a special case of the former, and
therefore, we focus on the former.
We define and solve two versions of the {\em detection game with fully active
attackers}, corresponding to two different formulations of the
problem: the Neyman--Pearson formulation and the Bayesian formulation.
%
%
%
In contrast to \cite{BT13},  here the players are
allowed to adopt randomized strategies. Specifically, the defender adopts
{\em randomized decision} strategies, while in
\cite{BT13} the defender's strategies were
confined to deterministic decision rules.
As for the attack,
it consists of the application of a {\em channel},
whereas in  \cite{BT13} it was confined to the application of a
deterministic function.
Moreover, the partially active case of \cite{BT13} can easily be obtained as a
special case of the fully active case considered here.
The problem of solving the game and then finding the optimum detector in the adversarial setting is not trivial and may not be possible in general.
Thus, we limit the complexity of the problem and make the analysis tractable by confining the decision to depend on a given set of
statistics of the observation.  Such an assumption, according to which the detector has access to a limited
 set of empirical statistics of the sequence,
is referred to as {\it limited resources} assumption (see \cite{merhav_sabbag_08} for an introduction on this terminology).
In particular, as done in \cite{BT13,BTtit},
we  limit the
detection resources to first order statistics, which are, as is well known,
sufficient statistics for memoryless systems \cite[Section 2.9]{CandT}.
While the sources are indeed assumed memoryless, one might still be concerned regarding the
sufficiency of first order statistics, in our setting, since
the attack channel
is not assumed memoryless in the first place. Adopting, nonetheless, the limited--resources
assumption to first order statistics,
is motivated mainly by its simplicity, but with the
understanding that the results can easily be extended to deal with arbitrarily higher
order empirical statistics as well. Moreover, an important bonus of this framework is that it
allows us to obtain fairly strong results concerning the game between the
defender and the attacker, as will be described below.

One of the main results of this paper
is the characterization of
an attack strategy which is both {\em dominant} (i.e.,
optimal no matter what the defence strategy is), and {\em universal}, i.e.,
independent of the (unknown) underlying sources.
Moreover, this optimal attack is
the same for both the Neyman-Pearson and Bayesian
games.
This result continues to hold also
for the partially active
case, thus creating a significant difference relative to previous works,
where the existence of a dominant strategy was
established regarding the defender only.


Some of our results
(in particular, the derivation of the equilibrium point for both
 the Neyman--Pearson and the Bayesian games),
have already appeared mostly without proofs in \cite{TBM_wifs15}.
Here we provide the full proofs of the
main theorems,
evaluate the payoff at
equilibrium for both the Neyman--Pearson and Bayesian games and
include the analysis of the ultimate performance of the games. Specifically,
%
we characterize the so called indistinguishability region (to be defined formally in Section \ref{sec.SA_limitingPerf}), namely
the set of the sources for which it is not possible to attain strictly positive exponents for both false positive and false negative probabilities under the Neyman-Pearson and the Bayesian settings.
Furthermore, the
setup and analysis presented in \cite{TBM_wifs15}
is extended by considering a more general case in which the
maximum allowed distortion levels
the attacker may introduce under the two hypotheses are different.

The paper is organized as follows.
In Section \ref{sec.not&def}, we
establish the notation and introduce the main concepts. In Section \ref{sec.DG_not_sym},
we formalize the problem and define the detection game with a fully active
adversary for both the Neyman-Pearson and the Bayesian
games, and then prove
the existence of a dominant and universal attack strategy.
The complete analysis of the Neyman-Pearson and Bayesian detection games,
namely, the study of the equilibrium point of the game and the
computation of the payoff at the equilibrium, are carried
out in
Sections \ref{sec.SA_NP} and \ref{sec.SA_B}, respectively.
Finally, Section \ref{sec.SA_limitingPerf} is devoted to
the analysis of the best achievable performance of the defender
and the characterization of the source distinguishability.

\section{Notation and definitions}
\label{sec.not&def}

Throughout the paper, random variables will be denoted by capital
letters
and specific realizations will be denoted by the
corresponding lower case letters.
All random variables that
denote signals in the system, will be assumed to have the same finite
alphabet, denoted by $\calA$.
Given a random variable $X$ and a
positive integer $n$,
we denote by
$\bX = (X_1,X_2,...,X_n)$, $X_i \in \mathcal{A}$,
$i=1,2,\ldots,n$,
a sequence of $n$ independent copies of $X$. According to the
above--mentioned notation rules, a specific realization of
$\bX$  is denoted by $\bx=(x_1,x_2,\ldots,x_n)$.
Sources will be denoted by the letter $P$. Whenever necessary, we will
subscript $P$ with the name of the relevant random variables:
given a random variable $X$,
$P_X$ denotes its
probability mass function (PMF). Similarly,
$P_{XY}$ denotes the joint PMF
of a pair of random variables, $(X,Y)$.
For two positive sequences,
$\{a_n\}$ and $\{b_n\}$, the notation $a_n \exe b_n$ stands for exponential
equivalence, i.e.,
$\lim_{n\to\infty} 1/n \ln\left(a_n/b_n\right)=0$, and $a_n \lexe b_n$
designates that $\limsup_{n\to\infty} 1/n \ln\left( a_n/b_n\right)\leq 0$.
For a given real $s$, we denote $[s]_+\dfn\max\{s,0\}$. We use notation $U(\cdot)$ for the Heaviside step function.

The type of a sequence
${\bx} \in \calA^n$ is defined as the empirical probability
distribution $\hat{P}_{\bx}$, that is,
the vector $\{\hat{P}_{\bx}(x), ~ x \in \mathcal{A}\}$ of
the relative frequencies of the various alphabet symbols in $\bx$.
A type class $\calT(\bx)$ is defined as the set of all sequences
having the same type as ${\bx}$.
When we wish to emphasize
the dependence of $\calT(\bx)$ on $\hat{P}_{\bx}$, we will use the
notation $\calT(\hat{P}_{\bx})$. Similarly, given a pair of sequences
$(\bx, \by)$, both of length $n$,
the joint type class $\calT(\bx, \by)$ is the set of sequence pairs $\{(\bx',\by')\}$ of length $n$
having the same empirical joint probability distribution
(or joint type) as $(\bx,\by)$,
$\hat{P}_{\bx \by}$, and
the conditional type class $\calT(\by|\bx)$ is the set of
sequences $\{\by'\}$  with $\hat{P}_{\bx\by'}=\hat{P}_{\bx\by}$.

Regarding information
measures,
the entropy associated with $\hat{P}_{\bx}$, which
is the empirical entropy of ${\bx}$, is denoted by
$\hat{H}_{\bx}(X)$. Similarly,
$\hat{H}_{\bx \by}(X,Y)$  designates the
empirical joint entropy of ${\bx}$ and $\by$,
and $\hat{H}_{\bx \by}(X|Y)$ is the conditional joint entropy.
We denote by $\calD(P \| Q)$
the Kullback--Leibler (K-L) divergence between
two sources, $P$ and $Q$ with the same alphabet (see \cite{CandT}).

Finally, we use letter $A$ to denote an attack channel; accordingly,
$A(\by| \bx)$ is the conditional probability of the channel output $\by$ given the channel input
$\bx$.
Given a permutation-invariant distortion function\footnote{
A permutation--invariant distortion function, $d(\bx,\by)$, is
a distortion function that is invariant if the same
permutation is applied to both $\bx$ and $\by$.}
$d:\calA^n\times\calA^n\to\reals^+$ and a maximum distortion
$\Delta$, we define the class $\calC_{\Delta}$ of admissible channels
$\{A(\by|\bx),~\bx,\by\in\calA^n\}$
as those that
assign zero
probability
to  every $\by$ with $d(\bx,\by) > n\Delta$.\\

\subsection{Basics of Game Theory}
\label{sec.intro_GT}

For the sake of completeness,
we introduce some basic definitions and concepts of game theory.
A two--player game is defined
as a quadruple $(\SS_1,\SS_2,u_1, u_2)$,
where $\SS_1 = \{s_{1,1} \dots s_{1,n_1}\}$ and
$\SS_2 = \{s_{2,1} \dots s_{2,n_2}\}$ are the
sets of strategies from which the first and the second player can choose, respectively, and $u_l(s_{1,i}, s_{2,j}), l= 1,2$,
is the payoff of the game for player $l$, when the first player
chooses the strategy $s_{1,i}$ and the second one chooses
$s_{2,j}$. Each player aims at maximizing its payoff function.
A pair of strategies $(s_{1,i}, s_{2,j})$ is called a {\it
profile}.
When $u_1(s_{1,i}, s_{2,j}) + u_2(s_{1,i}, s_{2,j}) = 0$,
the game is said to be a {\it zero-sum game}. For such games, the payoff
of the game $u(s_{1,i}, s_{2,j})$ is usually defined by adopting the
perspective of one of the two players: that is, $u(s_{1,i}, s_{2,j}) =
u_1(s_{1,i}, s_{2,j}) = - u_2(s_{1,i}, s_{2,j})$ if the defender's perspective
is adopted or vice versa.
The sets $\SS_1$, $\SS_2$
and the payoff functions are assumed known to both players.
In addition, we consider {\it strategic games}, i.e., games in which
the players choose their strategies ahead of
time, without knowing the strategy chosen by the opponent.

A common goal in game theory is to determine the existence of {\it
equilibrium points}, i.e. profiles that in {\em some sense}
represent a {\em satisfactory} choice for both players \cite{Osb94}.
The most famous notion of equilibrium is due to Nash \cite{Nash50}.
A profile is said to be a {\it Nash equilibrium}
if no player can improve its payoff by changing its strategy unilaterally.

Despite its popularity, the practical meaning of Nash equilibrium is often unclear,
since there is no guarantee that the players will end up playing at the
Nash equilibrium.
A particular kind of games for which stronger forms of
equilibrium exist are the so called {\em dominance solvable} games \cite{Osb94}.
The concept of dominance-solvability is directly related to the notion of dominant
and dominated strategies.
In particular, a strategy is said to be {\it strictly dominant} for
one player if it is the best strategy for this player, i.e.,
the strategy  that maximizes the payoff,
no matter
what the strategy of
the opponent may be.
In a similar way, we say that a strategy $s_{l,i}$ is
{\it strictly dominated} by strategy $s_{l,j}$, if the payoff achieved by
player $l$ choosing $s_{l,i}$ is always lower than that obtained by playing
$s_{l,j}$, regardless of the strategy of the other player.
Recursive elimination of dominated strategies
is a common technique for solving games. In the first step,
all the dominated strategies are removed from the set of available strategies,
since no {\it rational} player\footnote{In game theory, a rational player is supposed to act in a way that maximizes its payoff.}
would ever use them. In this way, a new,
smaller game is obtained. At this point, some strategies that were not
dominated before, may become dominated in the
new, smaller version of the game, and hence are eliminated as well. The process
goes on until no dominated strategy exists for  either player.
A {\em rationalizable equilibrium} is any profile which survives the
iterated elimination of dominated strategies \cite{ChenGames,Bern84}.
If at the end of the process only one profile is left, the remaining profile is said to be the
{\em only rationalizable equilibrium} of the game, which is also the only Nash equilibrium point.
Dominance solvable games are easy to analyze since, under the assumption of
rational players, we can anticipate that the players will choose the strategies
corresponding to the unique rationalizable equilibrium.
Another, related,  interesting notion of equilibrium is that of {\it dominant equilibrium}.
A dominant equilibrium is a profile which corresponds to dominant strategies
for both players and is the strongest kind of equilibrium that a strategic game may have.\\

\section{Detection Game with Fully Active Attacker}
\label{sec.DG_not_sym}


\subsection{Problem formulation}

Given two discrete memoryless sources, $P_0$ and $P_1$,
defined over
a common finite alphabet $\mathcal{A}$, we denote by
$\bx = (x_1,\ldots,x_n)\in \mathcal{A}^n$
a sequence emitted by one of these
sources. The sequence $\bx$ is available to the attacker.
Let $\by = (y_1,y_2,...,y_n) \in \mathcal{A}^n$
denote the sequence observed by the
defender:
when an attack occurs under both  $\calH_0$ and $\calH_1$, 
the observed sequence $\by$ is obtained as the output of an attack channel
fed by $\bx$.

In principle, we must distinguish between two cases: in the first, the attacker
is aware of the underlying hypothesis (hypothesis-aware attacker),
whereas in the second case it is not
(hypothesis-unaware attacker).
%
In the hypothesis-aware case,
the attack strategy is defined by two different conditional probability
distributions, i.e., two different attack channels: $A_0(\by|\bx)$, applied when
$\calH_0$ holds, and $A_1(\by|\bx)$, applied under $\calH_1$.
Let us denote by $Q_i(\cdot)$
the PMF of
$\by$ under  $\calH_i$,$i=0,1$. The attack induces the
following  PMFs on $\by$:
$Q_0(\by) = \sum_{\bx} P_0(\bx) A_0(\by|\bx)$ and
$Q_1(\by) = \sum_{\bx} P_1(\bx)
A_1(\by|\bx)$.

Clearly, in the hypothesis-unaware case, the attacker will apply the same
channel under $\calH_0$ and $\calH_1$, that is, $A_0=A_1$, and we will denote the common attack channel simply by $A$.
Throughout the paper, we focus on the
hypothesis-aware case as in view of this formalism, the
hypothesis-unaware case is just a special case.

Regarding
the defender,
we assume a
randomized decision strategy,
defined by $\Phi(\calH_i|\by)$, which
designates the probability of deciding in favor of $\calH_i$, $i=0,1$, given
$\by$.
Accordingly, the probability of a {\it false positive} (FP) decision error is given by
\begin{equation}
\label{P_FP}
P_{\mbox{\tiny FP}}(\Phi,A_0) = \sum_{\by}Q_0(\by)\Phi(\calH_1|\by),
\end{equation}
and similarly, the {\it false negative} (FN)
probability assumes the form:
\begin{equation}
\label{P_FN}
P_{\mbox{\tiny FN}}(\Phi,A_1) = \sum_{\by}Q_1(\by)\Phi(\calH_0|\by).
\end{equation}
\begin{figure}[t!]
\centering \includegraphics[width =0.48 \paperwidth]{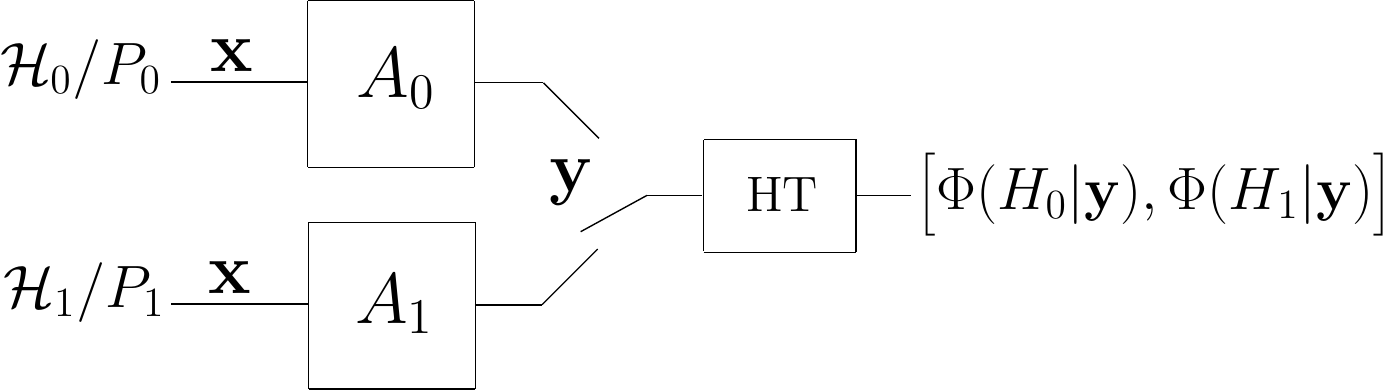}  
\caption{Schematic representation of the adversarial setup considered in this paper. In the case of  partially active attacker, channel $A_0$ corresponds to the identity channel.}
\label{fig.ADVsetup}
\end{figure}

Figure \ \ref{fig.ADVsetup} provides a
block diagram of the system with a fully active
attacker.
Obviously, the
partially active case, where
no attack occurs under $\calH_0$, can be seen as a
degenerate case of the fully active one, where $A_0$ is the identity channel $I$.
%
As in \cite{BT13}, due to the limited resources
assumption, the defender makes a decision based
on first order empirical statistics of $\by$,
which implies that
$\Phi(\cdot|\by)$ depends on ${\by}$ only via its type class
$\calT(\by)$.
%

Concerning the attack,
in order to limit the amount of distortion, we
assume a distortion constraint.
In the hypothesis--aware case, we allow the attacker different distortion
levels, $\Delta_0$ and $\Delta_1$, under $\calH_0$ and $\calH_1$,
respectively.
Then,
$A_0\in\mathcal{C}_{\Delta_0}$ and
$A_1\in\mathcal{C}_{\Delta_1}$,
where, for simplicity, we assume that a common (permutation-invariant) distortion function $d(\cdot,\cdot)$ is adopted in the two cases.

\subsection{Definition of the Neyman--Pearson and Bayesian Games}
\label{sec.def_game-NP-B}

One of the difficulties associated with the fully active
setting is that,
in the presence of a fully active attacker,
both the FP and FN probabilities depend on the attack channels.
We therefore consider
two different approaches which
lead to different formulations of the
detection game: in the first, the detection game is based on the Neyman-Pearson
criterion,
and in the second one, the Bayesian
approach is adopted.

For the Neyman-Pearson setting,
we define the game by assuming that
the defender adopts a conservative approach
and imposes an FP constraint pertaining to
the worst--case attack under $\calH_0$.
%
%
\begin{definition}
\label{def.NPgame}
The Neyman-Pearson detection game
is a zero-sum, strategic game defined as follows.
\begin{itemize}
  \item
The set $\mathcal{S}_{D}$ of strategies allowed to the defender
is the
class
of randomized decision rules $\{\Phi\}$ that satisfy
\begin{description}
  \item[(i)]
$\Phi(\calH_0|\by)$ depends on $\by$ only via its type.
  \item[(ii)] $\max_{A_0 \in \calC_{\Delta_0}} P_{\mbox{\tiny FP}}(\Phi,A_0)
\le e^{-n\lambda}$ for a prescribed constant
$\lambda > 0$, independent of $n$.
\end{description}
  \item
The set  $\calS_A$  of strategies allowed to the attacker
is the class  of
pairs of attack channels $(A_0,A_1)$ such that
$A_0 \in \calC_{\Delta_0}$, $A_1 \in \calC_{\Delta_1}$; that is,
$\calS_A=\calC_{\Delta_0}\times\calC_{\Delta_1}$.

    \item The payoff of the game is $u(\Phi,A_1)= P_{\mbox{\tiny FN}}(\Phi,A_1)$;
the attacker is in the quest of maximizing
$u(\Phi,A_1)$ whereas the defender wishes to minimize it.
\end{itemize}
\end{definition}
%
In the above definition, we require that the FP
probability decays exponentially fast with $n$, with
 an exponential rate {\em at least} as large as $\lambda$.
%
In the case of partially--active attack
(see the formulation in \cite{TBM_wifs15}), the FP probability
does not depend on the attack
but on the defender only;
accordingly, the constraint imposed by the defender in the above formulation
becomes $P_{\mbox{\tiny FP}}(\Phi) \le e^{-n\lambda}$.
Regarding the attacker, we have $\calS_A \equiv \calC_{0}
\times \calC_{\Delta_1}$, where $\calC_{0}$ is a singleton that
contains the identity channel only.

Another version of the detection game
is defined by assuming that the
defender follows a less conservative approach, that is, the Bayesian approach,
and tries to minimize a particular Bayes risk.
%
\begin{definition}
The Bayesian detection 
game is a zero-sum, strategic game defined as follow.
\begin{itemize}
\label{def.Bgame}
  \item The set $\mathcal{S}_{D}$ of strategies allowed to the defender is
the class
of the randomized decision rules $\{\Phi\}$
where $\Phi(\calH_0|\by)$ depends on $\by$ only via its type.
  \item The set $\mathcal{S}_{A}$ of strategies allowed
to the attacker is $\calS_A=\calC_{\Delta_0}\times\calC_{\Delta_1}$.
  \item The payoff of the game is
\begin{equation}
\label{payoff_soft_bayesian}
u(\Phi, (A_0,A_1)) = P_{\mbox{\tiny FN}}(\Phi,A_1) + e^{a n} P_{\mbox{\tiny FP}}(\Phi,A_0),
\end{equation}
for some constant $a$, independent of $n$.
\end{itemize}
\end{definition}
We observe that, in the definition of the payoff, the parameter $a$ controls the
tradeoff between
the two terms in
the exponential scale; whenever possible, the optimum defence strategy is
expected to yield error
exponents that differ exactly by $a$, so as to balance the contributions of
the two terms of \eqref{payoff_soft_bayesian}.

%
%
Notice also that, by defining the payoff as in
\eqref{payoff_soft_bayesian}, we are implicitly
considering for the defender only the strategies
$\Phi(\cdot|\by)$ such that $P_{\mbox{\tiny FP}}(\Phi,A_0) \lexe e^{- a n}$.
%
In fact,
any strategy that does not satisfy
this inequality yields a payoff $u > 1$, that cannot
be optimal, as it can be improved by always deciding
in favor of $\calH_0$ regardless of $\by$ ($u=1$).

As in \cite{BT13}, we focus on the asymptotic behavior of the game
as $n$ tends to infinity.
In particular, we are interested in
the FP and FN exponents defined as:
\begin{equation}
    \varepsilon_{\mbox{\tiny FP}} = - \limsup_{n \rarrow \infty} \frac{\ln P_{\mbox{\tiny FP}} (\Phi, A_0)}{n}; \quad \varepsilon_{\mbox{\tiny FN}} = - \limsup_{n \rarrow \infty} \frac{\ln P_{\mbox{\tiny FN}} (\Phi, A_1)}{n}.
\label{eq.err_exp}
\end{equation}
%
%

We say that a strategy is {\em asymptotically optimum} (or {\em dominant})
if it
is optimum (dominant) with respect to the
asymptotic exponential decay rate (or the exponent, for short)
of the payoff. 

\subsection{Asymptotically Dominant and Universal Attack}
\label{sec.universal_attack}

In this subsection, we characterize an
attack channel that, for both games, is asymptotically dominant and universal, in the sense
of being independent of the unknown underlying sources.
This result paves the way to
the solution of the two games.

Let $u$ denote a generic payoff function of the form
\begin{equation}
\label{Bayesian_payoff}
u = \gamma P_{\mbox{\tiny FN}}(\Phi, A_1) + \beta P_{\mbox{\tiny FP}}(\Phi,A_0),
\end{equation}
where $\beta$ and
$\gamma$ are given positive constants, possibly dependent on $n$.

We notice that the payoff of the Neyman-Pearson and Bayesian games defined in the previous section can be obtained  as particular cases: specifically,  $\gamma = 1$ and $\beta = 0$ for the Neyman-Pearson game and  $\gamma = 1$ and $\beta = e^{a n}$  for the Bayesian one.

%
\begin{theorem}
\label{Theorem_sym_A}
\label{theo_attack_H0_H1}
Let $c_n(\bx)$ denote the reciprocal of the total number of conditional type
classes
$\{\calT(\by|\bx)\}$
that satisfy the constraint $d(\bx,\by) \le n\Delta$
for a given $\Delta > 0$,
namely,
admissible conditional type classes\footnote{From the method
of the types it is known that $1\ge c_n(\bx)\ge (n+1)^{-|\calA|\cdot(|\calA|-1)}$ for any
$\bx$ \cite{CandT}.}.

Define:
\begin{equation}
\label{dominant_attack_channel}
A_{\Delta}^*(\by|\bx) = \left\{\begin{array}{ll}
\frac{c_n(\bx)}{|\calT(\by|\bx)|} & d(\bx,\by)\le n\Delta\\
0 & \mbox{elsewhere}\end{array}\right..
\end{equation}
%
%
Among all pairs of channels $(A_0,A_1) \in \calS_A$,
the pair $(A^*_{\Delta_0}, A_{\Delta_1}^*)$
minimizes the asymptotic exponent of $u$
for every $P_0$, $P_1$, every $\gamma, \beta \ge 0$ and
every permutation--invariant $\Phi(\calH_0|\cdot)$.
\end{theorem}

\begin{proof}
We first focus on the attack under $\calH_1$ and therefore on the FN probability.

Consider an arbitrary channel $A_1 \in \calC_{\Delta_1}$.
Let $\Pi:\calA^n\to \calA^n$ denote a permutation operator that
permutes any member of $\calA^n$ according to a given permutation matrix
%
and let
\begin{equation}
A_\Pi(\by|\bx)\dfn A_1(\Pi \by|\Pi \bx).
\end{equation}
Since the distortion function is assumed
permutation--invariant, the
channel $A_\Pi(\by|\bx)$ introduces the same
distortion as $A_1$ and hence satisfies the distortion constraint.
Due to the memorylessness of $P_1$ and the
assumption that $\Phi(\calH_0|\by)$ belongs to $\calS_D$, we have:
\begin{eqnarray}
P_{\mbox{\tiny
FN}}(\Phi, A_\Pi)&=&\sum_{\bx,\by}P_1(\bx)A_\Pi(\by|\bx)\Phi(\calH_0|\by)\nonumber\\
&=&\sum_{\bx,\by}P_1(\bx)A_1(\Pi\by|\Pi\bx)\Phi(\calH_0|\by)\nonumber\\
&=&\sum_{\bx,\by}P_1(\Pi\bx)A_1(\Pi\by|\Pi\bx)\Phi(\calH_0|\Pi\by)\nonumber\\
&=&\sum_{\bx,\by}P_1(\bx)A_1(\by|\bx)\Phi(\calH_0|\by)\nonumber\\
&=&P_{\mbox{\tiny FN}}(\Phi, A_1),
\end{eqnarray}
and so, $P_{\mbox{\tiny FN}}(\Phi,A_1)=P_{\mbox{\tiny FN}}(\Phi,\bar{A})$
where we have defined
\begin{equation}
\bar{A}(\by|\bx)=\frac{1}{n!}\sum_{\Pi}A_\Pi(\by|\bx)
=\frac{1}{n!}\sum_{\Pi}A_1(\Pi\by|\Pi\bx),
\end{equation}
which also introduces the same distortion as $A_1$.
Now, notice that this channel
assigns the same conditional probability to all sequences in the
same conditional type class $\calT(\by|\bx)$.
To see why this is true,
we observe that any sequence ${\by}' \in \calT(\by|\bx)$ can be
seen as being obtained from ${\by}$ through
the application of a permutation $\Pi'$ which leaves $\bx$ unaltered. Then, we have:
\begin{align}
\bar{A}(\by'|\bx) & = \bar{A}(\Pi' \by|\Pi' \bx)
= \frac{1}{n!}\sum_{\Pi} A_1(\Pi(\Pi' \by)|\Pi (\Pi' \bx))  \nonumber\\
 & =\frac{1}{n!}\sum_{\Pi}A_1(\Pi\by|\Pi\bx) = \bar{A}(\by|\bx).
\end{align}
Therefore, since $\bar{A}(\calT(\by|\bx)|\bx)\le 1$,
we argue that
\begin{align}
\bar{A}(\by|\bx)\lexe &
\left\{\begin{array}{ll}
\frac{1}{|\calT(\by|\bx)|} & d(\bx,\by)\le n\Delta\\
0 & \mbox{elsewhere}\end{array}\right. \nonumber\\
= & \frac{A_{\Delta_1}^*(\by|\bx)}{c_n(\bx)}\nonumber\\ \le &
(n+1)^{|\calA|\cdot(|\calA|-1)}A_{\Delta_1}^*(\by|\bx),
\end{align}
which implies that,
for every permutation--invariant
defence strategy $\Phi$,
%
\begin{equation}
P_{\mbox{\tiny FN}}(\Phi,A_1)\le
(n+1)^{|\calA|\cdot(|\calA|-1)}P_{\mbox{\tiny
FN}}(A_{\Delta_1}^*,\Phi)
\end{equation}
or equivalently
\begin{equation}
P_{\mbox{\tiny FN}}(\Phi,A_{\Delta_1}^*)\ge (n+1)^{-|\calA|\cdot(|\calA|-1)}P_{\mbox{\tiny
FN}}(A_1,\Phi).
\end{equation}
We conclude that $A_{\Delta_1}^*$
minimizes the error exponent of $P_{\mbox{\tiny FN}}(\Phi,A_1)$
across all channels in
$\calC_{\Delta_1}$ and for every $\Phi \in \calS_{D}$, regardless
of $P_1$.


A similar argument applies to the FP
probability to derive the optimum channel under $\calH_0$;
that is, from the memorylessness of $P_0$ and the
permutation--invariance of $\Phi(\calH_1|\cdot)$, we have:
\begin{equation}
\label{rel_2}
P_{\mbox{\tiny FP}}(\Phi,A_{\Delta_0}^*)\ge (n+1)^{-|\calA|\cdot(|\calA|-1)}P_{\mbox{\tiny
FP}}(A_0,\Phi),
\end{equation}
for every $A_0 \in \calC_{\Delta_0}$.
Accordingly, $A_{\Delta_0}^*$ minimizes the error exponent of $P_{\mbox{\tiny FP}}(\Phi,A_0)$.

We then have:
%
\begin{align}
\label{asympt_equality_bayes}
& \gamma P_{\mbox{\tiny FN}}(\Phi, A_1)   + \beta P_{\mbox{\tiny FP}}(\Phi,A_0)  \nonumber\\
& \hspace{0.3cm} \le (n+1)^{|\calA|\cdot(|\calA|-1)} (\gamma P_{\mbox{\tiny FN}}(\Phi, A_{\Delta_1}^*) + \beta P_{\mbox{\tiny FP}}(\Phi,A_{\Delta_0}^*)) \nonumber\\
& \hspace{0.3cm} \doteq \gamma P_{\mbox{\tiny FN}}(\Phi,A_{\Delta_1}^*) + \beta P_{\mbox{\tiny FP}}(\Phi,A_{\Delta_0}^*), 
\end{align}
for every $A_0 \in \calC_{\Delta_0}$ and $A_1 \in \calC_{\Delta_1}$. Notice
that, since the asymptotic equality is defined in the
logarithmic scale, eq. \eqref{asympt_equality_bayes} holds
no matter what the values of $\beta$ and $\gamma$ are, including values
that depend on $n$.
Hence, the pair of channels
$(A^*_{\Delta_0}, A_{\Delta_1}^*)$ minimizes the
asymptotic exponent of $u$ for any permutation--invariant
decision rule $\Phi(\calH_0|\cdot)$ and for any $\gamma, \beta \ge 0$.
\end{proof}
%


%
According to Theorem \ref{theo_attack_H0_H1}, for every zero-sum game with payoff function of the form in
\eqref{Bayesian_payoff}, if $\Phi$ is permutation-invariant, the pair of attack channels which is the most favorable to the attacker is $(A_{\Delta_0}^* , A_{\Delta_1}^*)$, which does not depend on $\Phi$. Then, the optimum attack strategy $(A_{\Delta_0}^* , A_{\Delta_1}^*)$ is {\em dominant}.
Specifically, given ${\bx}$, in order to generate $\by$ which causes a detection error
with the prescribed maximum allowed distortion, the attacker cannot do any better than
randomly selecting an admissible
conditional type class  according to the uniform distribution and
then choose at random $\by$ within this
conditional type class.
%
%
\begin{figure}[t!]
\centering \includegraphics[width = 0.48\paperwidth]{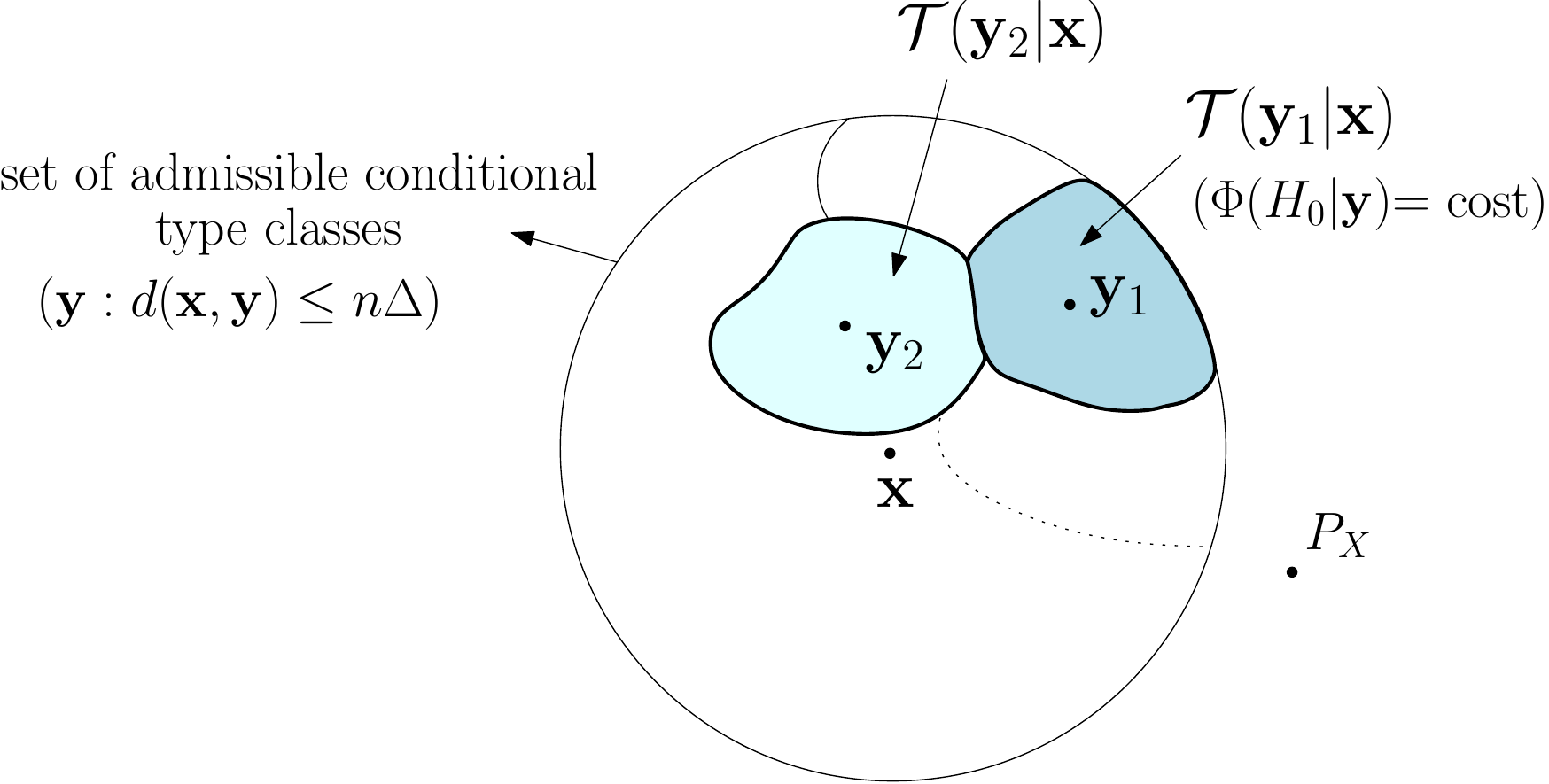} 
\caption{Graphical interpretation of the behavior of the attack channel $A^*_{\Delta}$.}
\label{fig.Opt_Attk}
\end{figure}
Figure \ref{fig.Opt_Attk} illustrates the intuition behind
the definition of the attack channel in \eqref{dominant_attack_channel}: 
since the number of conditional type classes is only polynomial in $n$,
the random choice of the conditional type class does not affect
the exponent of the error probabilities; besides, since the decision is the same
for all sequences within the same conditional type class,
the choice of
$\by$ within that conditional type class is immaterial.


As an additional result,
Theorem \ref{theo_attack_H0_H1} states that,
whenever an adversary aims at maximizing a
payoff function of the form
\eqref{Bayesian_payoff}, and as long as the
defence strategy is confined to the analysis of the first order statistics,
the (asymptotically) optimum attack strategy is
{\em universal} w.r.t. the sources $P_0$ and $P_1$, i.e., it depends neither on $P_0$ nor on $P_1$.

Finally, if $\Delta_0 = \Delta_1 = \Delta$, the optimum attack
consists  of applying the same channel $A^*_{\Delta}$ regardless of the
underlying hypothesis and then the optimum attack strategy
is {\em fully-universal}: the attacker needs to
know neither the sources ($P_0$ and $P_1$), nor the underlying hypothesis.
In this case, it becomes immaterial whether the attacker is
aware or unaware of the true hypothesis.
%
%
As a consequence of this property,
in the hypothesis-unaware case, when the attacker applies the
same channel under both hypotheses, subject to a
fixed maximum distortion $\Delta$, the optimum channel remains $A^*_{\Delta}$.


As a final remark, according to Theorem \ref{theo_attack_H0_H1},
for the partially active case, there exists an (asymptotically)
dominant and universal attack channel.
This result marks a  considerable difference
 relative to the results of \cite{BT13},
where the optimum deterministic attack function is found
using the rationalizability argument, that is,
by exploiting the existence of a dominant defence strategy,
and hence it is neither dominant nor universal.



\section{The Neyman-Pearson Detection Game}
\label{sec.SA_NP}

In this section, we study the detection game with a
fully active attacker in the Neyman-Pearson setup
as defined in Definition \ref{def.NPgame}.
From the analysis of Section \ref{sec.universal_attack},
we already know that there exists a dominant attack strategy.
Regarding the defender, we will determine the asymptotically
optimum strategy
regardless of
the dominant pair of attack channels;
in particular, as will been seen in Lemma \ref{theo_D_opt} below,
an asymptotically dominant defense strategy can
be derived
from a detailed analysis of the FP constraint.
As a consequence, the Neyman-Pearson detection game has a dominant equilibrium.

\subsection{Optimal Detection and Game Equilibrium}
\label{sec.SA_NP_GameEq}

The following lemma characterizes the optimal detection strategy in the
Neyman-Pearson setting.
\begin{lemma}
\label{theo_D_opt}
For the Neyman-Pearson game of Definition 1, the defence strategy
\begin{align}
\label{optimum_defence}
\Phi^*(\calH_1|\by)\dfn & \exp\left\{-n\left[\lambda- \min_{\bx: d(\bx,\by) \le n\Delta_0} \calD(\hat{P}_{\bx}\|P_0)\right]_+\right\},
\end{align}
is asymptotically dominant for the defender.
\end{lemma}

%
%

%
The proof appears in Appendix \ref{Appx.theo_D_opt}.

We point out that when the attacker is partially--active,
%
it is known from \cite{TBM_wifs15} that the optimum defence strategy is
%
%
\begin{align}
\label{optimum_defence_PA}
\Phi^*(\calH_1|\by)\dfn & \exp\left\{-n\left[\lambda-  \calD(\hat{P}_{\by}\|P_0)\right]_+\right\}.
\end{align}
From \eqref{optimum_defence_PA}, it is easy
to argue that there exists a deterministic strategy, corresponding to the Hoeffding test
\cite{Hoeffding65}, which is asymptotically
equivalent to $\Phi^*(\calH_1|\by)$. This result is in line with the one
 in \cite{BT13} (Lemma 1), where the class of
defence strategies is confined to deterministic
decision rules.

Intuitively, the extension
from \eqref{optimum_defence_PA} to \eqref{optimum_defence} is explained as follows.
In the case of fully active attacker, the defender
is subject to a constraint on
the maximum FP probability over $\calS_A$,
that is, the set of the admissible channels
$A \in \calC_{\Delta_0}$ (see Definition \ref{def.NPgame}).
From the analysis of Section \ref{sec.universal_attack},
%
channel $A_{\Delta_0}^*$
minimizes the FP exponent
over this set.
In order to satisfy the constraint
for a given sequence $\by$, the defender must
handle the worst--case value (i.e., the minimum) of $\calD(\hat{P}_{\bx}\|P_0)$
over all the type
classes $\calT(\bx|\by)$ which satisfy the distortion constraint,
or equivalently,
all the sequences $\bx$ such that $d(\bx,\by) \le n\Delta_0$.

According to Lemma \ref{theo_D_opt},
the best defence strategy
is asymptotically dominant.
Also, since $\Phi^*$ depends on
$P_0$ only, and not on $P_1$, it is referred to as {\it semi--universal}.\\

Concerning the attacker,
since the payoff
is a special case of \eqref{Bayesian_payoff}
with $\gamma = 1$ and $\beta = 0$,
the optimum pair of attack channels
is given by Theorem \ref{theo_attack_H0_H1} and
corresponds to $(A_{\Delta_0}^*,A_{\Delta_1}^*)$.

The following comment is in order.
Since the payoff of the game is defined in terms of the FN probability only,
it is independent of $A_0 \in \calC_{\Delta_0}$.
Furthermore, since the
defender adopts a conservative approach to
 guarantee the FP constraint for every
$A_0$, the constraint is satisfied  for every $A_0$ and  therefore all channel pairs of the
form $(A_0, A_{\Delta_1}^*)$, $A_0 \in \mathcal{S}_A$, are equivalent in terms of the
payoff.
Accordingly, in the
hypothesis--aware case,
the attacker can employ  any admissible channel under $\calH_0$.
In the Neyman--Pearson setting,
the sole fact that
the attacker is  active under $\calH_0$
forces the defender to take countermeasures that make
the choice of $A_0$ immaterial.

Due to the existence of dominant strategies
for both players, we can immediately state the following theorem.
%
\begin{theorem}
\label{theo_equilibirumNP}
Consider the Neyman-Pearson detection game
of Definition \ref{def.NPgame}.
Let $\Phi^*$ and $(A_{\Delta_0}^*,A_{\Delta_1}^*)$
be the strategies defined in Lemma \ref{theo_D_opt}
and Theorem \ref{theo_attack_H0_H1}, respectively.
The profile
$(\Phi^*,(A_{\Delta_0}^*,A_{\Delta_1}^*))$ is
an asymptotically dominant equilibrium of the game.
\end{theorem}

\subsection{Payoff at the Equilibrium}

In this section, we derive the payoff of the Neyman-Pearson  game
at the equilibrium of Theorem \ref{theo_equilibirumNP}. To do this, we will assume
an additive distortion function, i.e.,
$d(\bx,\by)=\sum_{i=1}^nd(x_i,y_i)$.
In this case,  $d(\bx,\by)$
can be expressed as $\sum_{ij} n_{\bx\by}(i,j) d(i,j)$,
where $n_{\bx\by}(i,j)=n\hat{P}_{\bx\by}(i,j)$ denotes
the number of occurrences of the pair $(i,j) \in \calA^2$ in $(\bx,\by)$.
Therefore, the distortion constraint
regarding $A_0$ can be rewritten as
$\sum_{(i,j) \in\calA^2} \hat{P}_{\bx\by}(i,j) d(i,j) \le \Delta_0$.
A similar formulation holds for $A_1$.
%

Let us define
\begin{equation}
\label{eq.tilde_D}
{\tilde{\calD}_{\Delta}}^n(\hat{P}_{\by},P)
\dfn \min_{\{\hat{P}_{\bx|\by}: E_{\bx \by} d(X,Y)
\le\Delta\}}
\calD(\hat{P}_{\bx}\|P),
\end{equation}
%
where $E_{\bx \by}$  denotes the {\em empirical expectation}, defined as
%
\begin{equation}
E_{\bx \by} d(X,Y)  =  \sum_{(i,j) \in\calA^2} \hat{P}_{\bx\by}(i,j) d(i,j)
\end{equation}
and the minimization is carried out for a given $\hat{P}_{\by}$.
%
Accordingly, the strategy in \eqref{optimum_defence}
can be  rewritten as
\begin{equation}
\Phi^*(\calH_1|\by)\dfn  \exp\left\{-n\left[\lambda- \tilde{\calD}_{\Delta_0}^n(\hat{P}_{\by}\|P_0)\right]_+\right\}.
\end{equation}
When $n \rightarrow \infty$,
$\tilde{\calD}_{\Delta}^n$ becomes\footnote{Due to the the density of rational numbers on
the real line, the admissibility set
in \eqref{eq.tilde_D} is dense in  that of \eqref{eq.tilde_D_limit};
since the the divergence functional is  continuous, the sequence
$\{\tilde{\calD}_{\Delta}^n(\hat{P}_{\by},P)\}_{n\ge 1}$
tends to $\tilde{\calD}_{\Delta}(P_{Y},P)$ as $n \rightarrow \infty$.}
\begin{equation}
\label{eq.tilde_D_limit}
\tilde{\calD}_{\Delta}(P_{Y},P) \dfn \min_{\{P_{X|Y}: E_{XY} d(X,Y)\le\Delta\}}
\calD(P_{X}\|P),
\end{equation}
%
%
%
where $E_{XY}$ denotes
expectation
w.r.t.\ $P_{XY}$.

%

Definition \eqref{eq.tilde_D_limit} can be stated for any PMF
$P_Y$ in the probability simplex in $\mathds{R}^{|\calA|}$.
%
Note that the minimization problem in
\eqref{eq.tilde_D_limit} has a unique solution
as it is a convex program.

The function $\tilde{\calD}_{\Delta}$
will have an important role in the remaining part of the paper,
especially in
the characterization of the asymptotic behavior of the games.
To draw a parallelism, $\tilde{\calD}_{\Delta}$ plays a role similar to that of the Kullback--Leibler divergence $\calD$
in classical detection theory for the non-adversarial case.

The basic properties of the
functional $\tilde{\calD}_{\Delta}(P_{Y},P)$ are the following: (i) it
is continuous in $P_{Y}$;  (ii) it has convex level sets, i.e.,
the set $\{P_Y : \tilde{\calD}_{\Delta}(P_{Y},P) \le t\}$ is
convex for every $t \ge 0$.
Point (ii) is a consequence of the following property,
which will turn out to be useful for proving
some of the results in the sequel (in particular, Theorem \ref{theo_e.e_2}, \ref{theo_e.e_B} and also \ref{theorem_EMD_FA}).
\begin{property}
\label{property_convex}
The function $\tilde{\calD}_{\Delta}(P_{Y},P)$
is convex in $P_Y$ for every fixed $P$.
\end{property}
The proof follows  from
the convexity of the divergence functional (see Appendix \ref{Appx.property}).



%

Using the above definitions, the
equilibrium payoff
is given by the following theorem:
\begin{theorem}
\label{theo_e.e_2}
Let the Neyman-Pearson detection
game be as in Definition \ref{def.NPgame}.
Let $(\Phi^*,(A_{\Delta_0}^*,A_{\Delta_1}^*))$
be the equilibrium profile of Theorem \ref{theo_equilibirumNP}.
Then,\footnote{We make explicit the dependence on the parameter $\lambda$ in the notation of the error exponent, since this will turn to be useful in the sequel.}
\begin{align}
\label{error_expon_2}
\varepsilon_{\mbox{\tiny
FN}}(\lambda) = &- \lim_{n \rightarrow \infty} \frac{1}{n} \ln P_{\mbox{\tiny
FN}}(\Phi^*, A_{\Delta_1}^*)  \nonumber\\ = & \underset{P_Y: \tilde{\calD}_{\Delta_0}(P_Y, P_0) \le \lambda}{\min} \tilde{\mathcal{D}}_{\Delta_1}(P_Y, P_1).
\end{align}
\end{theorem}
%
%
The proof, which appears in Appendix \ref{Appx.theo_e.e_2}, is based on
Sanov's theorem \cite{Sanov,CSnow}, by exploiting the compactness of the set
$\{P_Y: \tilde{D}_{\Delta_0}(P_Y, P_0) \le \lambda \}$.

From Theorem \ref{theo_e.e_2}
it follows that $\varepsilon_{\mbox{\tiny
FN}}(\lambda) = 0$  whenever
there exists a PMF $P_Y$ inside the set $\{P_Y: \tilde{D}_{\Delta_0}(P_Y, P_0) \le \lambda \}$ with $\Delta_1$-limited expected distortion from $P_1$.
When this condition does not hold, $P_{\mbox{\tiny
FN}}(\Phi^*, A_{\Delta_1}^*) \rightarrow 0$ exponentially rapidly.

%

For a partially--active attacker,
the  error exponent in \eqref{error_expon_2} becomes
%
%
\begin{align}
\label{error_expon_1}
\varepsilon_{\mbox{\tiny
FN}}(\lambda) = & \underset{P_Y: {\calD}(P_Y, P_0) \le \lambda}{\min} \tilde{\mathcal{D}}_{\Delta_1}(P_Y, P_1).
\end{align}
%
It can be shown that the error exponent in \eqref{error_expon_1} is the same as
 the error exponent of Theorem 2 in \cite{BT13}
(and Theorem 2 in \cite{BT_SM}),
where deterministic strategies are considered for both the defender and
the attacker.
Such equivalence can be explained as follows. As already pointed, the optimum
defence strategy  in \eqref{optimum_defence_PA} and the deterministic rule
found in \cite{BT13} are asymptotically equivalent  (see the discussion
immediately after Lemma \ref{theo_D_opt}). Concerning the attacker, even in
the more general setup (with randomized strategies) considered here, an asymptotically optimum attack could be derived as in  \cite{BT13}, that is, by considering the best response to the dominant defence strategy in  \cite{BT13}.
Such attack consists  of minimizing the divergence w.r.t. $P_0$, namely  $\calD(\hat{P}_{\bf y}||P_0)$, over all the admissible sequences ${\bf y}$, and then is deterministic. 
%
Therefore, concerning the partially active case, the asymptotic behavior of the game is equivalent to the one in \cite{BT13}. The main difference
between the setup in \cite{BT13} and the more general
one addressed in this paper relies on the {\em kind} of game equilibrium,
which is stronger here
(namely, a {\em dominant} equilibrium) due to the existence of dominant
strategies for both the defender and the attacker,
rather than for the defender only.

When the distortion function $d$ is a metric, we can state the following result, whose
proof appears in Appendix \ref{Appx.theo_e_fn_d}.
\begin{theorem}
\label{theo_e_fn_d}
When the distortion function $d$ is a metric,
eq.\ \eqref{error_expon_2} can be rephrased
as
\begin{align}
\label{error_expon_2_var}
\varepsilon_{\mbox{\tiny
FN}}(\lambda) = &\underset{P_Y: \calD(P_Y \| P_0) \le
\lambda}{\min}\tilde{\mathcal{D}}_{\Delta_0 + \Delta_1}(P_Y, P_1).
\end{align}
%
\end{theorem}
Comparing  eq. \eqref{error_expon_2_var} with
\eqref{error_expon_1}
is insightful for understanding the difference
between the
fully active and partially active cases.
%
Specifically, the FN error exponents  of both cases are the same when
the distortion under $\calH_1$ in the partially-active case is $\Delta_0 + \Delta_1$
(instead of $\Delta_1$).

When $d$ is not a metric, \eqref{error_expon_2_var}
is only an upper bound on  $\varepsilon_{\mbox{\tiny
FN}}(\lambda)$, as can be seen from the proof of Theorem \ref{theo_e_fn_d}.
Accordingly,
in the general case ($d$ is not a metric), applying distortion levels $\Delta_0$ and $\Delta_1$ to sequences from, respectively, $\calH_0$ and $\calH_1$ (in the fully active setup) is more favorable to the attacker
with respect to applying a distortion $\Delta_0 + \Delta_1$ to sequences from $\calH_0$ only (in the partially active setup).

\section{The Bayesian detection game}
\label{sec.SA_B}


In this section, we study the
Bayesian game (Definition \ref{def.Bgame}).
In contrast to the
Neyman--Pearson
game, in the Bayesian game, the optimal defence strategy
is found by assuming that
the strategy played by the attacker, namely the optimum
pair of channels $(A_0^*,A_1^*)$ of Theorem \ref{theo_attack_H0_H1},
is known to the defender, that is, by
exploiting the rationalizability argument
(see Section \ref{sec.intro_GT}).
Accordingly, the resulting optimum strategy is not dominant,
and so, the  associated equilibrium
is weaker  compared to that of the Neyman--Pearson game.

\subsection{Optimum Defence and Game Equilibrium}

Since the payoff in \eqref{payoff_soft_bayesian} is
a special case of  \eqref{Bayesian_payoff} with $\gamma = 1$ and $\beta = e^{a n}$,
for any defence strategy $\Phi \in \calS_{D}$,
the asymptotically optimum attack channels under $\calH_0$ and $\calH_1$ are
given by Theorem \ref{theo_attack_H0_H1},
and correspond to the pair  $(A_{\Delta_0}^*, A_{\Delta_1}^*)$.
Then,
we can determine the best defence strategy by assuming
that the attacker will play $(A_{\Delta_0}^*,A_{\Delta_1}^*)$
and evaluating the best response
of the defender to this pair of channels.



Our solution for the Bayesian detection game is
given in the following theorem, whose proof appears in
Appendix \ref{Appx.theo.Bayesian}.

\begin{theorem}
\label{theo.Bayesian}
Consider the Bayesian detection game  of Definition \ref{def.Bgame}.
Let $Q_0^*(\by)$ and
$Q_1^*(\by)$ be the probability distributions 
induced by channels $A_{\Delta_0}^*$ and $A_{\Delta_1}^*$,
respectively.

Then,\footnote{We remind that $U(\cdot)$ denotes the Heaviside step function.}
\begin{equation}
\label{optimum_defence_SA-2_1}
\Phi^{\#}(\calH_1|\by) =
U\left(\frac{1}{n} \log
\frac{Q_1^*(\by)}{Q_0^*(\by)} - a\right)
\end{equation}
%
is the optimum defence strategy.

If, in addition,
the distortion measure is additive, the defence strategy
%
\begin{equation}
\label{optimum_defence_SA-2_2}
\Phi^{\dag}(\calH_1|\by) = U\left(\tilde{\calD}_{\Delta_0}^n(\hat{P}_{\by},P_0)-
\tilde{\calD}_{\Delta_1}^n(\hat{P}_{\by},P_1) - a\right)
\end{equation}
is asymptotically optimum.
\end{theorem}

It is  useful to provide
the
asymptotically optimum strategy, $\Phi^{\dag}$, in addition to the optimal one,
$\Phi^{\#}$,  for the following reason: while
$\Phi^{\#}$  requires the
non-trivial computation of the two probabilities $Q_1(\by)$ and $Q_0(\by)$,
the strategy $\Phi^{\dag}$, which leads to the same payoff asymptotically, is easier to implement
because of its single-letter form.

We  now state the following theorem.
\begin{theorem}
\label{theo_equilubrium_B}
Consider the Bayesian  game  of
Definition \ref{def.Bgame}.
Let $(A_{\Delta_0}^*,A_{\Delta_1}^*)$ be
the attack strategy of
Theorem \ref{theo_attack_H0_H1} and let $\Phi^{\#}$ and
$\Phi^{\dag}$ be the defence  strategies
defined, respectively, in \eqref{optimum_defence_SA-2_1} and \eqref{optimum_defence_SA-2_2}.
The profiles $(\Phi^{\#},(A_{\Delta_0}^*,A_{\Delta_1}^*))$
and $(\Phi^{\dag},(A_{\Delta_0}^*,A_{\Delta_1}^*))$ are asymptotic rationalizable
equilibria of the game.
\end{theorem}

The analysis in this section can be easily
generalized to any payoff function defined as in \eqref{Bayesian_payoff},
i.e., for any $\gamma, \beta \ge 0$.

Finally, we observe
that, the fact that the equilibrium found in the Bayesian case (namely, a
rationalizable equilibrium) is weaker with respect to the equilibrium derived for the
Neyman--Pearson game (namely, a dominant equilibrium) is a consequence of
the fact that the Bayesian game is defined in a less restrictive manner than the
Neyman--Pearson game. This is due to the conservative approach adopted in the latter:
while in the Bayesian game the defender cares about both FP and FN probabilities and their tradeoff, in the Neymam--Pearson game the defender does not care about the value of the FP probability provided that its exponent is larger than $\lambda$, which is automatically guaranteed by restricting the set of strategies. This restriction simplifies the game so that a dominant strategy can be found for the restricted game.

\subsection{Equilibrium Payoff}
\label{sec.payoff_B}

We now derive the equilibrium payoff of the Bayesian game.
%
As in the
Neyman--Pearson game, we  assume an additive distortion measure.
For simplicity, we focus on
the asymptotically optimum defence strategy $\Phi^{\dag}$.
We  have the following theorem.
\begin{theorem}
\label{theo_e.e_B}
Let the Bayesian detection
game be  as in Definition \ref{def.Bgame}.
Let $(\Phi^{\dag},(A_{\Delta_0}^*,A_{\Delta_1}^*))$ be the
equilibrium profile  of
Theorem \ref{theo_equilubrium_B}. The asymptotic exponential
rate of the equilibrium
Bayes payoff $u$ is given by
\begin{align}
\label{expon_B}
 - \lim_{n \rightarrow \infty} & \frac{1}{n}  \ln \left(u(\Phi^{\dag}, (A_{\Delta_0}^*, A_{\Delta_1}^*)) \right)=  \nonumber\\
&  \min_{P_Y} \left(\max\left\{
\tilde{\calD}_{\Delta_1}(P_Y,P_1), (\tilde{\calD}_{\Delta_0}(P_Y,P_0) - a)\right\}\right).
\end{align}
\end{theorem}
The proof appears in Appendix \ref{Appx.theo.e.e.B}.

According to  Theorem \ref{theo_e.e_B},
the asymptotic exponent of $u$ is zero
if there exists a PMF $P_Y^*$
with $\Delta_1$-limited expected distortion from $P_1$
such that  $\tilde{\calD}_{\Delta_0}(P_Y^*,P_0) \le a$.
%
%
%
Therefore, when we focus on the case of zero asymptotic exponent of the payoff,
the parameter $a$ plays a role similar
to $\lambda$ in the Neyman--Pearson game.
%
%
By further inspecting the exponent expressions of Theorems \ref{theo_e.e_B} and  \ref{theo_e.e_2},
we observe that,  when $a=\lambda$,
the exponent in \eqref{expon_B} is smaller than
or equal to the one in \eqref{error_expon_2},
where
equality holds only when both \eqref{expon_B} and \eqref{error_expon_2}
vanish.
%
%
However, comparing these two cases in the general case is difficult
because of the different definition of the payoff functions and, in particular, the different role taken by the parameters $\lambda$ and $a$.
In the Neyman--Pearson game, in fact, the payoff corresponds to the FN probability and is not affected by the value of the FP probability, provided that its exponent is larger than $\lambda$; in this way, the ratio between FP and FN error exponent at the equilibrium is generally smaller than $\lambda$ (a part for the case in which the asymptotic exponent of the payoff is zero).
In the Bayesian case,
the payoff is a weighted combination of the two types of errors and then the term with the largest exponent is the dominating term, namely, the one which determines the asymptotic behavior; in this case, the parameter $a$ determines the exact tradeoff between the FP and FN exponent in the equilibrium payoff.

\section{Source Distinguishability}
\label{sec.SA_limitingPerf}

%

In this section, we investigate the performance
of the Neyman--Pearson and Bayesian games as functions of
$\lambda$ and $a$ respectively.
From the expressions of the equilibrium payoff exponents,
it is clear that
the Neyman--Pearson and the Bayesian payoffs increase as $\lambda$
and $a$ decrease, respectively.
In particular, by setting $\lambda=0$ and $a=0$, we obtain the
largest achievable payoffs of both games which correspond to the best achievable performance for the
defender.
%
%
Therefore, we say that two sources are {\em distinguishable}
under the Neyman--Pearson (resp.\ Bayesian)
setting,
if there exists a value of $\lambda$  (resp.\
$\alpha$) such that
the FP and FN exponents at the equilibrium of the game are simultaneously strictly positive.
When such a condition does not hold, we say that
the sources are indistinguishable. 
Specifically, in this section, we characterize, under both the Neyman--Pearson and the Bayesian settings, the {\em indistinguishability region}, defined as the set of the alternative sources
that cannot be distinguished from a given source $P_0$,
given the attack distortion levels $\Delta_0$ and $\Delta_1$.
Although each game has a different asymptotic behavior,
we will see that the indistinguishability regions in the
Neyman--Pearson and the Bayesian settings are
the same.
%
The study of the distinguishability between
the sources under adversarial conditions, performed in this section, in a way extends the Chernoff-Stein
lemma \cite{CandT} to the adversarial setup (see \cite{BT_SM}).
%


We start by proving the following result for the Neyman--Pearson game.
\begin{theorem}
\label{theorem_EMD_FA}
Given two memoryless sources $P_0$ and $P_1$ and distortion
levels $\Delta_0$ and $\Delta_1$, the maximum achievable
FN exponent for the
Neyman--Pearson game is:
\begin{align}
\label{best_e_e_NP}
\underset{\lambda \rightarrow 0}{\lim} \hspace{0.07cm} \varepsilon_{\mbox{\tiny
FN}}(\lambda)
\hspace{0.08cm}  =  \varepsilon_{\mbox{\tiny FN}}(0)  =  \underset{\{P_{Y|X}: E_{XY} d(X,Y)\le\Delta_0, \hspace{0.08cm} (P_{XY})_X = P_0\}}{\min} \tilde{\mathcal{D}}_{\Delta_1}(P_Y, P_1),
\end{align}
where $\varepsilon_{\mbox{\tiny
FN}}(\lambda)$ is
as in Theorem \ref{theo_e.e_2}.
\end{theorem}
%
%
%
The theorem is an immediate consequence of the continuity of $\varepsilon_{\mbox{\tiny FN}}(\lambda)$ as
$\lambda \rightarrow 0^+$, which
follows by the continuity of $\tilde{\mathcal{D}}_{\Delta}$ with respect to $P_Y$ and the density of the set $\{P_Y:
\tilde{\calD}_{\Delta_0}(P_Y, P_0) \le \lambda\}$ in $\{P_Y:
\tilde{\calD}_{\Delta_0}(P_Y, P_0) = 0\}$ as $\lambda \rightarrow 0^+$ \footnote{It holds true from Property \ref{property_convex}.}.


We notice that, if $\Delta_0 = \Delta_1 = 0$, there is only an admissible point in the set in \eqref{best_e_e_NP}, for which $P_Y = P_0$; then, $\varepsilon_{\mbox{\tiny FN}}(0) = {\calD}(P_0 || P_1)$, which corresponds to the best achievable FN exponent known from the classical literature for the non-adversarial case (Stein lemma \cite{CandT}, Theorem 11.8.3).

Regarding the Bayesian setting,
we  have the following theorem, the proof of which
appears in  Appendix
\ref{Appx.theorem_EMD_FA_B}.
\begin{theorem}
\label{theorem_EMD_FA_B}
Given two memoryless sources $P_0$ and $P_1$ and distortion levels
$\Delta_0$ and $\Delta_1$, the maximum achievable exponent of the
equilibrium Bayes payoff is
\begin{align}
 - \lim_{a \rightarrow 0} \lim_{n \rightarrow \infty} & \frac{1}{n}  \ln \left(u(\Phi^{\dag}, (A_{\Delta_0}^*, A_{\Delta_1}^*)) \right)=  \nonumber\\
&  \min_{P_Y} \left(\max\left\{
\tilde{\calD}_{\Delta_1}(P_Y,P_1), \tilde{\calD}_{\Delta_0}(P_Y,P_0)\right\}\right),
\label{best_e_e_B}
\end{align}
%
%
where the inner limit at the left hand side is  as defined in Theorem \ref{theo_e.e_B}.
\end{theorem}
%
%

Since $\tilde{\calD}_{\Delta_1}(P_Y,P_1)$, and similarly
$\tilde{\calD}_{\Delta_0}(P_Y,P_0)$, are convex functions of $P_Y$, and reach their minimum in $P_1$, resp. $P_0$,\footnote{The fact that $\tilde{\calD}_{\Delta_0}$ ($\tilde{\calD}_{\Delta_1}$) is 0 in a $\Delta_0$-limited ($\Delta_1$-limited) neighborhood of $P_0$ ($P_1$), and not just in $P_0$ ($P_1$), does not affect the argument.}
the minimum over $P_Y$ of the maximum between these quantities
(right-hand side of \eqref{best_e_e_B}) is attained when
$\tilde{\calD}_{\Delta_1}(P_Y^*,P_1) =  \tilde{\calD}_{\Delta_0}(P_Y^*,P_0)$,
for some PMF $P_Y^*$.
%
%
%
%
This resembles the best achievable exponent in the Bayesian
probability of error for the non-adversarial case, which is
attained when ${\calD} (P_Y^*\|P_0) = {\calD}(P_Y^*\|P_1)$ for some
$P_Y^*$  (see \cite{CandT}, Theorem 11.9.1).
In that case, from the expression of the divergence function,
such $P_Y^*$ is found in a closed form and the resulting exponent is equivalent to the
Chernoff information (see Section 11.9 in \cite{CandT}).


From Theorem \ref{theorem_EMD_FA} and \ref{theorem_EMD_FA_B},
it follows that there is no positive $\lambda$, res. $a$, for which the asymptotic exponent of the equilibrium payoff is strictly positive,
if there exists a PMF $P_Y$
such that the following conditions are both satisfied:
\begin{equation}
\label{system_cond}
\left\{\begin{array}{ll}
\tilde{\calD}_{\Delta_0}(P_Y, P_0)=0 & \\
\tilde{\calD}_{\Delta_1}(P_Y, P_1)=0.
\end{array}\right.
\end{equation}
In this case, then, $P_0$ and $P_1$ are indistinguishable
under both the Neyman--Pearson and the Bayesian settings.
We observe that the condition $\tilde{\calD}_{\Delta}(P_Y, P_X) = 0$ is equivalent to the
following:\footnote{For ease of notation, given a joint PMF $Q_{XY}$ with marginal PMFs $P_X$ and $P_Y$,
we use notation $(Q_{XY})_Y = P_Y$ (res. $(Q_{XY})_X = P_X$) as short for $\sum_x Q_{XY}(x,y) =
P_Y(y)$, $\forall y \in \calA$ (res. $\sum_y Q_{XY}(x,y) =
P_X(x)$, $\forall x \in \cal{A}$).}
\begin{equation}
\label{Wasserstein_dist}
\min_{Q_{XY}:\tiny{\substack{ (Q_{XY})_X = P_X\\ (Q_{XY})_Y = P_Y}}} E_{XY} d(X,Y) \le \Delta,
\end{equation}
where the expectation $E_{XY}$ is w.r.t $Q_{XY}$.
In computer vision applications,
the left-hand side of \eqref{Wasserstein_dist} is known as the {\em Earth
Mover Distance} (EMD)
 between $P_X$ and $P_Y$,
which is denoted by $\text{\em EMD}_{d}(P_X,P_Y)$ (or, by symmetry, $\text{\em EMD}_{d}(P_Y,P_X)$) \cite{RTG00}.
It is also known as the
$\rho$-bar distortion measure \cite{gray2011entropy}.

A brief comment concerning the analogy between the minimization in  \eqref{Wasserstein_dist}
and {\em optimal transport theory} is worth.
%
The minimization problem  in \eqref{Wasserstein_dist}
is known in the Operations Research literature as
{\em Hitchcock Transportation Problem} (TP) \cite{Hitchcock}. 
Referring to the original Monge formulation
of this problem \cite{monge1781}, $P_X$ and $P_Y$ can be interpreted as
two different ways of piling up a certain amount of soil; then,
$P_{XY}(x,y)$ denotes the quantity of soil shipped
from location (source) $x$ in $P_X$ to location (sink) $y$ in $P_Y$
and $d(x, y)$ is the cost for shipping a unitary amount of soil from $x$ to $y$.
In transport theory terminology, $P_{XY}$ is referred to as {\em transportation map}.
According to this perspective, evaluating the {\em EMD}
corresponds to finding the minimal transportation cost of moving a pile of soil into the other.
Further insights on this parallel can be found in \cite{BT_SM}.\\

We summarize our findings in the following
corollary, which characterizes the conditions for
distinguishability under both the Neyman--Pearson and the Bayesian setting.
\begin{corollary}[Corollary  to Theorems \ref{theorem_EMD_FA} and \ref{theorem_EMD_FA_B}]
\label{cor_Gamma}
Given a memoryless source $P_0$ and
distortion levels $\Delta_0$ and $\Delta_1$, the set of the PMFs that
cannot be distinguished from $P_0$ in both the Neyman--Pearson and
Bayesian settings is given by
\begin{equation}
\label{A_winning_region}
\Gamma= \left\{P :
\min_{P_Y : \text{\em EMD}_{d}(P_Y,P_0) \le \Delta_0} \text{\em EMD}_{d}(P_Y,P) \le \Delta_1\right\}.
\end{equation}
\end{corollary}
%
Set $\Gamma$ is the indistinguishability region. By definition (see the beginning of this section),
%
%
the PMFs inside $\Gamma$ are those for which, as a consequence of the attack, the FP and FN probabilities cannot go to zero simultaneously with strictly positive exponents.
Clearly, if $\Delta_0 = \Delta_1 = 0$, that is, in the non-adversarial case, $\Gamma = \{P_0\}$, as any two distinct sources are always distinguishable.


When $d$ is a metric, for a given $P \in \Gamma$, the computation of
the optimum $P_Y$ can be traced back to the computation of the
{\em EMD} between $P_0$ and $P$, as stated by the following corollary,
whose proof appears in Appendix \ref{Appx.cor_Gamma_d_metric}.
\begin{corollary}[Corollary  to Theorems \ref{theorem_EMD_FA} and \ref{theorem_EMD_FA_B}]
\label{cor_Gamma_d_metric}
When $d$ is a metric, given the source $P_0$ and distortion levels
$\Delta_0$ and $\Delta_1$, for any fixed $P$, the minimum  in \eqref{A_winning_region}
is achieved when
\begin{equation}
\label{optimumPY}
P_Y = \alpha P_0 + (1-\alpha)P, \quad \alpha = 1 - \frac{\Delta_0}{\text{\em EMD}(P_0,P)}.
\end{equation}
Then, the set of PMFs that cannot be distinguished from $P_0$
in the Neyman--Pearson and Bayesian setting is given by
\begin{equation}
\label{A_winning_region_d}
\Gamma = \{P : \text{\em EMD}_d(P_0, P) \le \Delta_0 + \Delta_1\}.
\end{equation}
\end{corollary}
According to Corollary \ref{cor_Gamma_d_metric},
when $d$ is a metric,
the performance of the game depends only on the  sum of
distortions,
$\Delta_0 + \Delta_1$, and it is immaterial how this amount is
distributed between the two  hypotheses.

In the general case ($d$ not a metric), the condition on the {\em EMD} stated in \eqref{A_winning_region_d}
is sufficient in order for $P_0$ and $P$ be indistinguishable, that is
%
$\Gamma \supseteq \{P : \text{\em EMD}_d(P_0, P) \le \Delta_0 + \Delta_1\}$ (see discussion in
Appendix \ref{Appx.cor_Gamma_d_metric}, at the end of the proof of Corollary \ref{cor_Gamma_d_metric}).
%
%
Furthermore,
in the case of an $L_p^p$
distortion function ($p \ge 1$),  i.e., $d(\bx, \by) = \sum_{i = 1}^n |x_i - y_i|^p$, we have the following
corollary.
%
\begin{corollary}[Corollary  to Theorems \ref{theorem_EMD_FA} and \ref{theorem_EMD_FA_B}]
\label{cor_L_p_p}
When $d$ is the $L_p^p$ distortion function, for some $p \ge 1$,
the set $\Gamma$ can be bounded as follows
%
\begin{equation}
\label{A_winning_region_L_2}
\Gamma \subseteq \{P :
\text{\em EMD}_{L_p^p}(P_0, P) \le ({\Delta_0}^{1/p} + {\Delta_1}^{1/p})^p\}.
\end{equation}
\end{corollary}
Corollary \ref{cor_L_p_p} can be proven by exploiting
the H{\"o}lder inequality \cite{keyInequalities} (see Appendix \ref{Appx.cor_L_p_p}).

\section{Conclusions}
\label{sec.conc}

We considered the problem of binary hypothesis testing when an attacker is active under both hypotheses,
and then an attack is carried out aiming at
both false negative and false positive errors. By modeling the defender-attacker interaction as a game, we
defined and solved two different detection games: the Neyman--Pearson and the Bayesian game.
This paper extends the analysis in \cite{BT13}\cite{BT13}, where the attacker is
active under the alternative hypothesis only. Another aspect of greater
generality is that here both players are allowed to use randomized
strategies.
By relying on the method of types, the main result of
this paper is the existence of an attack strategy which is both {\em dominant} and {\em universal},
that is, optimal regardless of the statistics of the sources.
The optimum attack strategy is also independent of the underlying hypothesis,
namely {\em fully-universal}, when the distortion introduced by the attacker in the two cases is the same.
From the analysis of the asymptotic behavior of the equilibrium payoff
we are able to establish conditions under which the sources can be reliably distinguished in the fully-active adversarial setup.
The theory developed permits
to assess the security of the detection in adversarial setting and
give insights on how the detector should be designed
in such a way to make the attack hard.

Among the possible directions for future work, we mention the extension to
multiple hypothesis
testing.
Another interesting direction is the extension  to continuous  alphabets,
which calls for an extension of the method of types to this case,
or to more realistic models of finite alphabet sources, still amenable to
analysis, like
Markov sources. 
As mentioned in the introduction, it would be also relevant to overcome the limitation to first order statistics, by extending the analysis to  higher order statistics and getting equilibria in a similar fashion.
Finally, we mention the case of unknown sources, where the sources
are estimated from training data, possibly corrupted by the attacker.
In this scenario, the detection game has been studied for a partially active case, with both uncorrupted and corrupted training data \cite{BTtit,BTtit18}. The extension of such analyses to the  fully active scenario considered in this paper is a further interesting direction for future research.

\section*{Acknowledgment}

We thank Alessandro Agnetis of
the University of Siena, for the useful discussions on optimization concepts underlying the
computation of the {\em EMD}.



\numberwithin{equation}{section}


\appendices



\section{Neyman--Pearson detection game}
\label{App-NPgame}

This appendix contains the proofs of the results in Section \ref{sec.SA_NP}.

\subsection{Proof of Lemma \ref{theo_D_opt}}
\label{Appx.theo_D_opt}

\renewcommand{\theequation}{\thesection.\arabic{equation}}

Whenever existent, the dominant defence strategy can be obtained by solving:
\begin{equation}
\label{strategy_Defence}
\min_{\Phi \in \calS_{D}} P_{\mbox{\tiny FN}}(\Phi, A_1),
\end{equation}
for any attack channel $A_1$.
%
%
Below, we first show that $P_{\mbox{\tiny FN}}(\Phi^*,A_1)\lexe P_{\mbox{\tiny FN}}(\Phi,A_1)$ for every $\Phi \in \calS_{D}$ and for every $A_1$, that is, $\Phi^*$ is asymptotically dominant.
Then, by proving that $\max_{A \in \calC_{\Delta_0}} P_{\mbox{\tiny FP}}(\Phi^*,A)$ fulfills the FP constraint, we show that $\Phi^*$ is also admissible. Therefore, we can conclude that $\Phi^*(\cdot|\by)$ asymptotically solves \eqref{strategy_Defence}.
%
Exploiting the memorylessness of $P_0$
and the permutation invariance of $\Phi(\calH_1|\by)$ and $d(\bx,\by)$, for every $\by^\prime\in\calA^n$ we
have,
\begin{align}
\label{derivation_opt_D}
e^{- \lambda n} \ge & \max_A \sum_{\bx,\by}P_0(\bx)A(\by|\bx)\Phi(\calH_1|\by)\nonumber\\
\ge & \sum_{\by} \left(\sum_{\bx} P_0(\bx)A_{\Delta_0}^*(\by|\bx) \right) \Phi(\calH_1|\by)\nonumber\\
= & \sum_{\by} \left(\sum_{\bx: d(\bx,\by) \le n\Delta_0} P_0(\bx) \cdot \frac{c_n(\bx)}{|\calT(\by|\bx)|} \right) \Phi(\calH_1|\by)\nonumber\\
\ge & (n + 1)^{-|\calA|\cdot(|\calA| - 1)} \sum_{\by} \left(\sum_{\bx: d(\bx,\by) \le n\Delta_0}  \cdot \frac{P_0(\bx)}{|\calT(\by|\bx)|} \right) \Phi(\calH_1|\by)\nonumber\\
{\stackrel{(a)}{\ge}} & (n + 1)^{-|\calA|\cdot(|\calA| - 1)}  |\calT(\by')| \left(\max_{\bx: d(\bx,\by') \le n\Delta_0} |\calT(\bx|\by')| \cdot \frac{P_0(\bx)}{|\calT(\by'|\bx)|} \right) \Phi(\calH_1|\by') \nonumber\\
{\stackrel{(b)}{=}} & (n + 1)^{-|\calA|\cdot(|\calA| - 1)} \Phi(\calH_1|\by') \max_{\bx: d(\bx,\by') \le n\Delta_0}  P_0(\bx) \cdot |\calT(\bx)|   \nonumber\\
\ge & \Phi(\calH_1|\by') \max_{\bx: d(\bx,\by') \le n\Delta_0}  \frac{e^{-n\calD(\hat{P}_{\bx} \| P_0)} }{(n + 1)^{|\calA|^2\cdot(|\calA| - 1)}} \nonumber\\
= & \Phi(\calH_1|\by') \frac{\exp\left\{-n  \min_{\bx: d(\bx,\by') \le n\Delta_0}  \calD(\hat{P}_{\bx} \| P_0)\right\} }{(n + 1)^{|\calA|^2 \cdot(|\calA| - 1)}},
\end{align}
%
where $(a)$ is due to the permutation invariance
of the distortion function $d$ and $(b)$
is due to the identity
$|\calT(\bx)|\cdot|\calT(\by|\bx)|\equiv
|\calT(\by)|\cdot|\calT(\bx|\by)|\equiv|\calT(\bx,\by)|$.

It now follows that
\begin{equation}
\Phi(\calH_1|\by) \lexe \exp\left\{- n \left[\lambda - \min_{\bx: d(\bx,\by) \le n\Delta_0}  \calD(\hat{P}_{\bx} \| P_0) \right]\right\}.
\end{equation}
Since $\Phi(\calH_1|\by)$ is a probability,
\begin{align}
\label{defence_proof_rel1}
 \Phi(\calH_1|\by)  & \lexe
\min\left\{1,\exp\left[- n \left(\lambda - \min_{\bx: d(\bx,\by) \le n\Delta_0}  \calD(\hat{P}_{\bx} \| P_0) \right)\right]\right\}\nonumber\\
& = \Phi^*(\calH_1|\by).
\end{align}
Consequently,
$\Phi^*(\calH_0|\by)\lexe \Phi(\calH_0|\by)$ for every $\by$, and so,
$P_{\mbox{\tiny FN}}(\Phi^*,A_1)\lexe P_{\mbox{\tiny FN}}(\Phi, A_1)$ for every $A_1$.
For convenience, let us denote
$$k_n(\by)=\lambda - \min_{\bx: d(\bx,\by) \le n\Delta_0}\calD(\hat{P}_{\bx}\|
P_0),$$
so that $ \Phi^*(\calH_1|\by) = \min\{1,e^{- n \cdot k_n(\by)}\}$.
We now show that $\Phi^*(\calH_1|\by)$ satisfies the
FP constraint, up to a polynomial term in $n$,
i.e., it satisfies the constraint asymptotically.
%
\begin{align}
\label{contraintFP_satisfaction}
\max_{A \in \calC_{\Delta_0}} P_{\mbox{\tiny FP}}(\Phi^*,A) & \le (n+1)^{|\calA|\cdot(|\calA| - 1)} P_{\mbox{\tiny FP}}(\Phi^*,A^*) \nonumber\\ & = (n+1)^{|\calA|\cdot(|\calA| - 1)} \sum_{\bx,\by}P_0(\bx)A_{\Delta_0}^*(\by|\bx)\Phi^*(\calH_1|\by)\nonumber\\
& = (n+1)^{|\calA|\cdot(|\calA| - 1)}  \sum_{(\bx, \by): d(\bx,\by) \le n\Delta_0} P_0(\bx) \cdot \frac{c_n(\bx)}{|\calT(\by|\bx)|} \cdot \Phi^*(\calH_1|\by)\nonumber\\
& \le (n+1)^{|\calA|\cdot(|\calA| - 1)}  \sum_{(\bx, \by): d(\bx,\by) \le n\Delta_0}  \frac{P_0(\bx)}{|\calT(\by|\bx)|} \cdot \Phi^*(\calH_1|\by)\nonumber\\
& \le (n+1)^{2 |\calA|\cdot(|\calA| - 1)}  \sum_{\by}  \left( \max_{\bx: d(\bx,\by) \le n\Delta_0}  |\calT(\bx|\by)| \cdot \frac{P_0(\bx)}{|\calT(\by|\bx)|} \right) \Phi^*(\calH_1|\by) \nonumber\\
& = (n+1)^{2 |\calA|\cdot(|\calA| - 1)} \left( \sum_{\hat{P}_{\by}:  k_n(\by) \ge 0} e^{- n k_n(\by)} \left( \max_{\bx: d(\bx,\by) \le n\Delta_0}  |\calT(\bx)| \cdot P_0(\bx) \right) + \right. \nonumber\\
& \hspace{4cm} \left. + \sum_{\hat{P}_{\by}:  k_n(\by) < 0}  \left( \max_{\bx: d(\bx,\by) \le n\Delta_0}  |\calT(\bx)| \cdot P_0(\bx) \right) \right) \nonumber\\
& \le (n+1)^{2 |\calA|\cdot(|\calA| - 1)} \left( \sum_{\hat{P}_{\by}:  k_n(\by) \ge 0} e^{-n \lambda} + \right. \nonumber\\
& \hspace{4cm} \left. + \sum_{\hat{P}_{\by}:  k_n(\by) < 0}  \exp\left\{-n \min_{\bx: d(\bx,\by) \le n\Delta_0}\calD(\hat{P}_{\bx} \| P_0)\right\} \right) \nonumber\\
& \le (n+1)^{(|\calA|^2 + 2 |\calA|)\cdot(|\calA| - 1) + |\calA|} e^{-n \lambda}.
\end{align}

\subsection{Proof of Property \ref{property_convex}}
\label{Appx.property}

We  next prove that
for any two PMFs $P_{Y_1}$ and $P_{Y_2}$ and any
$\lambda \in (0,1)$,
%
\begin{equation}
\label{conv_relation}
\tilde{\calD}_{\Delta}(\lambda P_{Y_1} + (1 - \lambda) P_{Y_2},P) \le \lambda \tilde{\calD}_{\Delta}(P_{Y_1}, P) +  (1 - \lambda) \tilde{\calD}_{\Delta}(P_{Y_2},P).
\end{equation}
%
%

Let us rewrite $\tilde{\calD}_{\Delta}$ in
\eqref{eq.tilde_D_limit} by expressing the minimization
in terms of the joint PMF $P_{XY}$:
\begin{equation}
\tilde{\calD}_{\Delta}(P_{Y},P) \dfn \min_{\{Q_{XY}: E_{XY} d(X,Y) \le \Delta, (Q_{XY})_Y = P_Y\}}
\calD\big( (Q_{XY})_X\|P \big),
\end{equation}
%
where we used  $(Q_{XY})_Y = P_Y$  as short for $\sum_{x}Q_{XY}(x,y) = P_Y(y)$,
$\forall y$,
and
we made explicit the dependence of $\calD(P_{X}\|P)$ on $Q_{XY}$.
Accordingly:
\begin{equation}
\label{gen_div_convex}
\tilde{\calD}_{\Delta}(\lambda P_{Y_1} + (1-\lambda) P_{Y_2},P) = \min_{\{Q_{XY}: E_{XY} d(X,Y) \le \Delta, \hspace{0.1cm}(Q_{XY})_Y = \lambda P_{Y_1} + (1-\lambda) P_{Y_2}\}}
\calD\big( (Q_{XY})_X \|P\big).
\end{equation}
%
We find convenient to rewrite the right-hand side of \eqref{gen_div_convex}  by minimizing over
pairs of PMFs $(Q_{XY}',Q_{XY}'')$ and considering the convex combination of these PMFs with weights $\lambda$ and $(1-\lambda)$, in place of $Q_{XY}$; hence
%
%
%
%
\begin{equation}
\label{gen_div_convex2}
\tilde{\calD}_{\Delta}(\lambda P_{Y_1} + (1-\lambda) P_{Y_2},P) =  \underset{(Q_{XY}', Q_{XY}'') \in \mathcal{H}}{\min}
\calD\big( \lambda  ( Q_{XY}')_X +  (1-\lambda) ( Q_{XY}'')_X \|P \big),
\end{equation}
where
\begin{align}
\label{gen_div_convex2_H}
 \mathcal{H} = & \left\{(Q_{XY}', Q_{XY}''): \lambda (Q_{XY}')_Y + (1-\lambda) (Q_{XY}'')_Y = \lambda P_{Y_1} + (1-\lambda) P_{Y_2}, \right.\nonumber\\
  & \hspace{7cm} \left.  \lambda E_{XY}' d(X,Y) + (1-\lambda) E_{XY}'' d(X,Y) \le \Delta\right\}.
\end{align}
Let
\begin{align}
 \mathcal{H}' = & \left\{Q_{XY}':   E_{XY}'' d(X,Y) \le \Delta, ( Q_{XY}')_Y = P_{Y_1}\right\} \times \left\{ Q_{XY}'': E_{XY}'' d(X,Y) \le \Delta, (Q_{XY}'')_Y = P_{Y_2}\right\};
\end{align}
then, $\mathcal{H}' \subset \mathcal{H}$, where the set $\mathcal{H}'$ is separable in $ Q_{XY}'$ and $ Q_{XY}'$.
Accordingly, \eqref{gen_div_convex2}-\eqref{gen_div_convex2_H} can be upper bounded by
%
\begin{equation}
\underset{Q_{XY}':   E_{XY}'' d(X,Y) \le \Delta, ( Q_{XY}')_Y = P_{Y_1}}{\min} \hspace{0.2cm}\underset{Q_{XY}'': E_{XY}'' d(X,Y) \le \Delta, (Q_{XY}'')_Y = P_{Y_2}}{\min}
\calD\big( \lambda  (Q_{XY}')_X +  (1-\lambda) ( Q_{XY}'')_X \|P \big).\\
\end{equation}
%
By the convexity of $\calD\big((Q_{XY})_X \| P\big)$ with respect to $Q_{XY}$\footnote{This is a consequence of the fact  that the divergence function is convex in its arguments and the operation $(\cdot)_X$ is linear (see  Theorem 2.7.2 in \cite{CandT}).},
it follows that
\begin{equation}
\calD\big( \lambda (Q_{XY}')_X +  (1-\lambda) (Q_{XY}'')_X \|P \big) \le \lambda  \calD\big(( Q_{XY}')_X\| P\big) + (1-\lambda) \calD\big( ( Q_{XY}'')_X \|P \big).
\end{equation}
Note that the above relation is not strict since it might be that $(Q_{XY}')_X = (Q_{XY}'')_X = P$.
Then, an upper bound for $\tilde{\calD}_{\Delta}(\lambda P_{Y_1} + (1-\lambda) P_{Y_2},P)$ is given by
\begin{equation}
 \underset{Q_{XY}': \sum_{x} Q_{XY}' = P_{Y_1},  E_{XY}'' d(X,Y) \le \Delta}{\min}   \lambda \calD\big((Q_{XY}')_X\| P\big) +  \underset{Q_{XY}'': (Q_{XY}'')_Y = P_{Y_2}, E_{XY}'' d(X,Y) \le \Delta}{\min} (1-\lambda)\calD\big( (Q_{XY}'')_X \|P \big),
\end{equation}
thus proving  \eqref{conv_relation}.

\subsection{Proof of Theorem \ref{theo_e.e_2}}
\label{Appx.theo_e.e_2}

We start by proving the upper bound for the FN probability:
\begin{align}
\label{derivation_opt_D}
P_{\mbox{\tiny
FN}}(\Phi^*, A_{\Delta_1}^*) =  & \sum_{\bx,\by}P_1(\bx)A_{\Delta_1}^*(\by|\bx)\Phi^*(\calH_0|\by)\nonumber\\
= & \sum_{\by}  \sum_{\bx: d(\bx,\by) \le n\Delta_1} P_1(\bx) \frac{c_n(\bx)}{|\calT(\by|\bx)|} \left(1 - e^{-n[\lambda-\tilde{\calD}_{\Delta_0}^n(\hat{P}_{\by}, P_0)]_+}\right)\nonumber\\
\le & \sum_{\by}  \sum_{\bx: d(\bx,\by) \le n\Delta_1} \frac{P_1(\bx)}{|\calT(\by|\bx)|} \left(1 - e^{-n[\lambda-\tilde{\calD}_{\Delta_0}^n(\hat{P}_{\by}, P_0)]_+}\right)\nonumber\\
= & \sum_{\by}  \sum_{\hat{P}_{\bx|\by}:  E_{\bx \by} d(X,Y) \le\Delta_1} |\calT(\hat{P}_{\bx|\by})| \frac{ e^{-n[\hat{H}_{\bx}(X)+
\calD(\hat{P}_{\bx}\|P_1)]}}{|\calT(\hat{P}_{\by|\bx})|}  \left(1 - e^{-n[\lambda-\tilde{\calD}_{\Delta_0}^n(\hat{P}_{\by}, P_0)]_+}\right)\nonumber\\
= & \sum_{\hat{P}_{\by}}  \sum_{\hat{P}_{\bx|\by}:  E_{\bx \by} d(X,Y) \le\Delta_1} |\calT(\hat{P}_{\bx})|  e^{-n[\hat{H}_{\bx}(X)+
\calD(\hat{P}_{\bx}\|P_1)]} \left(1 - e^{-n[\lambda-\tilde{\calD}_{\Delta_0}^n(\hat{P}_{\by}, P_0)]_+}\right)\nonumber\\
= & \sum_{\hat{P}_{\by}}  \sum_{\hat{P}_{\bx|\by}:  E_{\bx \by} d(X,Y)\le\Delta_1} e^{-n
\calD(\hat{P}_{\bx}\|P_1)} \left(1 - e^{-n[\lambda-\tilde{\calD}_{\Delta_0}^n(\hat{P}_{\by}, P_0)]_+}\right)\nonumber\\
= & \sum_{\hat{P}_{\by}: \tilde{\calD}_{\Delta_0}^n(\hat{P}_{\by} , P_0) < \lambda}  \sum_{\hat{P}_{\bx|\by}:  E_{\bx \by} d(X,Y) \le\Delta_1} e^{-n
\calD(\hat{P}_{\bx}\|P_1)} \left(1 - e^{-n(\lambda-\tilde{\calD}_{\Delta_0}^n(\hat{P}_{\by}, P_0))}\right)\nonumber\\
\le & \sum_{\hat{P}_{\by}: \tilde{\calD}_{\Delta_0}^n(\hat{P}_{\by} , P_0) < \lambda}  \sum_{\hat{P}_{\bx|\by}:  E_{\bx \by} d(X,Y) \le\Delta_1} e^{-n
\calD(\hat{P}_{\bx}\|P_1)}\nonumber\\
\le &  (n + 1)^{2|\calA|\cdot (|\calA| - 1)}   \exp\left\{- n \underset{\hat{P}_{\by}:\tilde{\calD}_{\Delta_0}^n(\hat{P}_{\by} , P_0) < \lambda}{\min} \left[ \underset{\hat{P}_{\bx|\by}: E_{\bx \by} d(X,Y) \le\Delta_1}{\min} \mathcal{D}(\hat{P}_{\bx} \|P_1) \right]\right\}\nonumber\\
\le &  (n + 1)^{2|\calA|\cdot (|\calA| - 1)}   \exp\left\{- n \underset{P_Y: \tilde{\calD}_{\Delta_0}(P_Y , P_0) < \lambda}{\inf} \left[ \underset{P_{X|Y}: E_{XY} d(X,Y) \le \Delta_1}{\min} \mathcal{D}(P_X \|P_1) \right]\right\}\nonumber\\
\le &  (n + 1)^{2|\calA|\cdot (|\calA| - 1)}   \exp\left\{- n \underset{P_Y: \tilde{\calD}_{\Delta_0}(P_Y, P_0) \le \lambda}{\min} \left[ \underset{P_{X|Y}: E_{XY} d(X,Y) \le \Delta_1}{\min} \mathcal{D}(P_X \|P_1) \right]\right\}.
\end{align}
%
%
%

Then:
\begin{equation}
\limsup_{n \rightarrow \infty} \frac{1}{n} \ln P_{\mbox{\tiny
FN}}(\Phi^*, A_{\Delta_1}^*) \le - \underset{P_Y: \tilde{\calD}_{\Delta_0}(P_Y , P_0) \le \lambda}{\min} \left[ \underset{P_{X|Y}: E_{XY} d(X,Y)  \le \Delta_1}{\min} \mathcal{D}(P_X \|P_1) \right].
\end{equation}
We now move on to the lower bound.
\begin{align}
\label{derivation_opt_D}
P_{\mbox{\tiny
FN}}(\Phi^*, A_{\Delta_1}^*) =  & \sum_{\bx,\by}P_1(\bx)A_{\Delta_1}^*(\by|\bx)\Phi^*(\calH_1|\by)\nonumber\\
\ge & (n+1)^{-|\calA|\cdot(|\calA|-1)}   \sum_{\by}  \sum_{\bx: d(\bx,\by) \le n\Delta_1} \frac{P_1(\bx)}{|\calT(\by|\bx)|} \left(1 - e^{-n[\lambda-\tilde{\calD}_{\Delta_0}^n(\hat{P}_{\by}, P_0)]_+}\right)\nonumber\\
= &  (n+1)^{-|\calA|\cdot(|\calA|-1)}   \sum_{\hat{P}_{\by} : \tilde{\calD}_{\Delta_0}^n(\hat{P}_{\by} , P_0) < \lambda}  \sum_{\hat{P}_{\bx|\by}:  E_{\bx \by} d(X,Y) \le\Delta_1}  e^{- n \mathcal{D}(\hat{P}_{\bx} \|P_1)} \left(1 - e^{-n(\lambda-\tilde{\calD}_{\Delta_0}^n(\hat{P}_{\by}, P_0))}\right)\nonumber\\
\ge & (n+1)^{-|\calA|\cdot(|\calA|-1)}    e^{- n \mathcal{D}(\hat{P}_{\bx} \|P_1)} (1 - e^{-n(\lambda-\tilde{\calD}_{\Delta_0}^n(\hat{P}_{\by}, P_0))}),
\end{align}
where, for a fixed $n$,  $\hat{P}_{\by}$ is a  PMF  that satisfies $\tilde{\calD}_{\Delta_0}^n(\hat{P}_{\by}, P_0) \le \lambda -  (\ln n)/ n$
and $\hat{P}_{\bx|\by}$ is such that the distortion constraint is satisfied.
Since the set of rational PMFs is dense in the probability simplex, two such sequences can be chosen
in such a way that $(\hat{P}_{\by}, \hat{P}_{\bx|\by}) \rightarrow (P_Y^*,P_{X|Y}^*)$,\footnote{We are implicitly exploiting the fact that set  $\{\tilde{\calD}_{\Delta_0}^n(\hat{P}_{\by}, P_0) < \lambda\}$ is dense in $\{\tilde{\calD}_{\Delta_0}(P_Y, P_0) \le \lambda\}$, for every $\lambda > 0$, which holds true from Property \ref{property_convex}.} where
\begin{equation}
 (P_Y^*,P_{X|Y}^*) = \underset{(P_Y,P_{X|Y})}{\arg \min} \hspace{0.2cm} \underset{P_Y: \tilde{\calD}_{\Delta_0}(P_Y, P_0) \le \lambda}{\min} \left[ \underset{P_{X|Y}: E_{XY} d(X,Y)  \le \Delta_1}{\min} \mathcal{D}(P_X \|P_1) \right].
\end{equation}
%
%
Therefore, we can assert that:
\begin{align}
\liminf_{n \rightarrow \infty} \frac{1}{n} \ln P_{\mbox{\tiny
FN}}(\Phi^*, A_{\Delta_1}^*) \ge &  \lim_{n \rightarrow \infty} \frac{1}{n} \ln \left[e^{- n \mathcal{D}(\hat{P}_{\bx} \|P_1)} \left(1 - e^{-n(\lambda- \tilde{\calD}_{\Delta_0}^n (\hat{P}_{\by}, P_0))}\right)\right] \nonumber \\
 = &-  \lim_{n \rightarrow \infty}  \mathcal{D}(\hat{P}_{\bx} \|P_1) \nonumber\\
  = &-   \mathcal{D}(P_X^* \|P_1) \nonumber\\
 = & - \underset{P_Y:  \tilde{\calD}_{\Delta_0}(P_Y \|P_0) \le \lambda}{\min} \left[ \underset{P_{X|Y}: E_{XY} d(X,Y) \le \Delta_1}{\min} \mathcal{D}(P_X \|P_1) \right].
\end{align}
By combining the upper and lower bounds, we conclude that
$\lim \sup$ and $\lim \inf$ coincide.
Therefore the limit of the sequence $1/n \ln P_{\mbox{\tiny
FN}}$ exists and the theorem is proven.

\subsection{Proof of Theorem \ref{theo_e_fn_d}}
\label{Appx.theo_e_fn_d}

First, observe that, by exploiting the definition of $\tilde{\calD}_{\Delta}$, \eqref{error_expon_2} can be rewritten
as
\begin{align}
\label{error_expon_2_var1}
\varepsilon_{\mbox{\tiny
FN}} =  & \underset{P_Y: \tilde{\calD}_{\Delta_0}(P_Y, P_0) \le \lambda}{\min} \left(\underset{P_{X|Y}: E_{XY} d(X,Y) \le \Delta_1}{\min} \mathcal{D}(P_X \| P_1)\right) \nonumber\\
 = & \underset{P_Z: \calD(P_Z \| P_0) \le \lambda}{\min} \hspace{0.1cm} \underset{P_{Y|Z}:E_{YZ} d(Y,Z) \le \Delta_0}{\min} \left(\underset{P_{X|Y}: E_{XY} d(X,Y) \le \Delta_1}{\min} \mathcal{D}(P_X \| P_1)\right).
\end{align}
%

To prove the theorem, we now show that \eqref{error_expon_2_var1} can be simplified as follows:
\begin{align}
\label{error_expon_2_var2}
\varepsilon_{\mbox{\tiny
FN}} = &  \underset{P_{Z}: \mathcal{D}(P_Z \|P_0) \le \lambda}{\min} \left( \underset{P_{X|Z} : E_{XZ}d(X,Z) \le \Delta_0 + \Delta_1}{\min} \mathcal{D}(P_X \|P_1)\right),
\end{align}
which is equivalent to
\eqref{error_expon_2_var} (see Section \ref{sec.DG_not_sym}).
The equivalence of the expressions in \eqref{error_expon_2_var1} and \eqref{error_expon_2_var2} follows from the equivalence of the two feasible sets for the PMF $P_X$. We first show that any feasible $P_X$ in \eqref{error_expon_2_var1} is also feasible in \eqref{error_expon_2_var2}. Let then $P_X$ be a feasible PMF in \eqref{error_expon_2_var1}.
By exploiting the properties of the triangular inequality property of the  distance, we have that, regardless of the specific choice of the distributions $P_{Y|Z}$ and $P_{X|Y}$ in \eqref{error_expon_2_var1},
\begin{equation}
E_{XZ}d(X,Z) \le E_{XYZ}[d(X,Y) + d(Y,Z)] =  E_{XY}d(X,Y) + E_{YZ}d(Y,Z) \le \Delta_0 + \Delta_1,
\end{equation}
and then $P_X$ is a feasible PMF in \eqref{error_expon_2_var2}.
To prove the opposite inclusion, we observe that, for any $P_Z$ and $P_{X|Z}$  such that $\mathcal{D}(P_Z \|P_0) \le \lambda$ and $E_{XZ} d(X,Z) \le \Delta_0 + \Delta_1$, it is possible to define a variable $Y$, and then two conditional PMFs $P_{Y|Z}$ and $P_{X|Y}$, such that $E_{XY} d(X,Y) \le \Delta_1$ and  $E_{YZ} d(Y,Z) \le \Delta_0$. To do so, it is sufficient to let $P_Y$ be the convex combination of $P_{X}$ and $P_Z$, that is $P_Y = \alpha P_X + (1-\alpha)P_Z$ where $\alpha = \Delta_0/(\Delta_0 + \Delta_1)$. With this choice for the marginal, we can define $P_{X|Y}$ so that $P_{XY}$ satisfies\footnote{By adopting the transportation theory perspective introduced towards the end of Section \ref{sec.SA_limitingPerf},  we can look at $P_X$ and $P_Y$ as two ways of piling up a certain amount of soil; then $P_{XY}$ can be interpreted as a map which moves $P_X$ into $P_Y$ ($P_{XY}(i,j)$ corresponds to the amount of soil moved from position $i$ to $j$). The map in \eqref{clever_joint} is the one which leaves in place a percentage $\alpha$ of the mass and moves the remaining $(1-\alpha)$ percentage to fill the pile $(1-\alpha)P_Z$ according to map $(1-\alpha) P_{XZ}$.}
\begin{align}
\label{clever_joint}
&P_{XY}(i,j) = (1-\alpha)P_{XZ}(i,j) \quad  \forall i, \forall j \neq i,\nonumber\\
&P_{XY}(i,i) = (1-\alpha)P_{XZ}(i,i) + \alpha P_X(i)  \quad \forall i;
\end{align}
similarly, $P_{Y|Z}$ can be chosen such that $P_{YZ}$ satisfies
\begin{align}
\label{clever_joint2}
&P_{YZ}(i,j) = \alpha P_{XZ}(i,j)   \quad \forall i, \forall j \neq i,\nonumber\\
&P_{YZ}(i,i) =  \alpha P_{XZ}(i,i) + (1-\alpha) P_Z(i)  \quad \forall i.
\end{align}
It is easy to see that, with the above choices, $E_{XY} d(X,Y) = (1-\alpha) E_{XZ}d(X,Z)$ and $E_{YZ} d(Y,Z) = \alpha E_{XZ}d(X,Z)$.
Then, $E_{XY} d(X,Y) \le  (1-\alpha)(\Delta_0 + \Delta_1) \le \Delta_1$ and $E_{YZ} d(Y,Z) \le \Delta_0$.
Consequently, any $P_X$ belonging to the set in \eqref{error_expon_2_var2} also belongs to the one in \eqref{error_expon_2_var1}.

\section{Bayesian detection game}
\label{App-Bgame}

This appendix contains the proofs for Section \ref{sec.SA_B}.

\subsection{Proof of Theorem \ref{theo.Bayesian}}
\label{Appx.theo.Bayesian}

Given the probability distributions $Q_0(\by)$ and
$Q_1(\by)$ induced by $A_{\Delta_0}^*$ and $A_{\Delta_1}^*$ respectively,
the optimum decision rule is deterministic and is given by the
likelihood ratio test (LRT) \cite{DetEst}:
\begin{equation}
\label{LRT}
\frac{1}{n} \ln
\frac{Q_1(\by)}{Q_0(\by)}\underset{\calH_0}{\overset{\calH_1}{\gtrless}} a,
\end{equation}
which proves the optimality of the decision rule in
\eqref{optimum_defence_SA-2_1}.

To prove the asymptotic optimality of the decision rule in \eqref{optimum_defence_SA-2_2}, let us approximate $Q_0(\by)$ and $Q_1(\by)$ using the method of types as follows:
\begin{eqnarray}
Q_0(\by)&=&\sum_{\bx}P_0(\bx)A_{\Delta_0}^*(\by|\bx)\nonumber\\
&\exe&\sum_{\bx:~d(\bx,\by)\le n\Delta_0} e^{-n[\hat{H}_{\bx}(X)+
\calD(\hat{P}_{\bx}\|P_0)]}\cdot e^{-n\hat{H}_{\bx\by}(Y|X)}\nonumber\\
&\exe&\max_{\bx:~d(\bx,\by)\le n\Delta_0} e^{n\hat{H}_{\bx\by}(X|Y)}\cdot
\left(e^{-n[\hat{H}_{\bx}(X)+
\calD(\hat{P}_{\bx}\|P_0)]} \right.\nonumber\\
& & \hspace{4.5cm} \left.\cdot e^{-n\hat{H}_{\bx\by}(Y|X)}\right)\nonumber\\
&=&\max_{\bx:~d(\bx,\by)\le n\Delta_0}
e^{-n[\hat{H}_{\by}(Y)+
\calD(\hat{P}_{\bx}\|P_0)]}\nonumber\\
&\stackrel{(a)} = &\exp\left\{-n\left[\hat{H}_{\by}(Y)+ \right. \right.\nonumber\\
& & \left.\left. \hspace{1.5cm} + \min_{\{\hat{P}_{\bx|\by}:E_{\bx \by} d(X,Y)\le\Delta_0\}}
\calD(\hat{P}_{\bx}\|P_0)\right]\right\}\nonumber\\
&=&\exp\left\{-n[\hat{H}_{\by}(Y)+
\tilde{\calD}_{\Delta_0}^n(\hat{P}_{\by},P_0)]\right\},
\end{eqnarray}
where in $(a)$ we exploited the additivity of the distortion function $d$.
Similarly,
\begin{equation}
Q_1(\by)\exe
\exp\left\{-n[\hat{H}_{\by}(Y)+\tilde{\calD}_{\Delta_1}^n(\hat{P}_{\by},P_1)]\right\}.
\end{equation}
Thus, we have the following asymptotic approximation to the LRT:
\begin{equation}
\tilde{\calD}_{\Delta_0}^n(\hat{P}_{\by},P_0)-
\tilde{\calD}_{\Delta_1}^n(\hat{P}_{\by},P_1)
\underset{\calH_0}{\overset{\calH_1}{\gtrless}} a,
\end{equation}
which proves the second part of the theorem.

\subsection{Proof of Theorem \ref{theo_e.e_B}}
\label{Appx.theo.e.e.B}

In order to make
 the expression of
$u(P_{\mbox{\tiny FN}}(\Phi^{\dag},(A_{\Delta_0}^*, A_{\Delta_1}^*)))$
explicit, let us first evaluate the two error probabilities at equilibrium.
%
%
Below, we derive the
lower and upper bound  on the probability of
$\by$ under $\calH_1$, when the attack channel is $A_{\Delta_1}^{*}$:
\begin{equation}
(n+1)^{-|\calA \|\calA - 1|} e^{-n[\hat{H}_{\by}(Y)+
\tilde{\calD}_{\Delta_1}^n(\hat{P}_{\by},P_1)]} \le Q_1^*(\by) < (n+1)^{|\calA|^2} e^{-n[\hat{H}_{\by}(Y)+
\tilde{\calD}_{\Delta_1}^n(\hat{P}_{\by},P_1)]}.
\end{equation}
The same bounds hold for $Q_0^*(\by)$,
with $\tilde{\calD}_{\Delta_0}$  replacing $\tilde{\calD}_{\Delta_0}$.
For the FN probability,
the upper bound  is
%
\begin{align}
\label{UB_P_FN}
P_{\mbox{\tiny FN}}(\Phi^{\dag},A_{\Delta_1}^*) = & \sum_{\by} Q_1^*(\by) \cdot \Phi^{\dag}(\calH_0|\by) \nonumber\\
= & \sum_{\by: \tilde{\calD}_{\Delta_0}^n(\hat{P}_{\by},P_0) - \tilde{\calD}_{\Delta_1}^n(\hat{P}_{\by},P_1) < a} Q_1^*(\by) \nonumber\\
\le & (n+1)^{|\calA|^2} \sum_{\by: \tilde{\calD}_{\Delta_0}^n(\hat{P}_{\by},P_0) - \tilde{\calD}_{\Delta_1}^n(\hat{P}_{\by},P_1) < a} e^{-n[\hat{H}_{\by}  +
\tilde{\calD}_{\Delta_1}^n(\hat{P}_{\by},P_1)]}\nonumber\\
\le & (n+1)^{|\calA|^2 + |\calA|} \max_{\hat{P}_{\by}: \tilde{\calD}_{\Delta_0}^n(\hat{P}_{\by},P_0) - \tilde{\calD}_{\Delta_1}^n(\hat{P}_{\by},P_1) < a} e^{-n
\tilde{\calD}_{\Delta_1}^n(\hat{P}_{\by},P_1)}\nonumber\\
=  &  (n+1)^{|\calA|^2 + |\calA|} \exp\left\{-n\left(\min_{\hat{P}_{\by}: \tilde{\calD}_{\Delta_0}^n(\hat{P}_{\by},P_0) - \tilde{\calD}_{\Delta_1}^n(\hat{P}_{\by},P_1) < a}
\tilde{\calD}_{\Delta_1}^n(\hat{P}_{\by},P_1)\right)\right\}.
\end{align}
Then,
\begin{align}
\label{limsup}
 - \limsup_{n\rightarrow \infty} \frac{1}{n} \ln(P_{\mbox{\tiny FN}}(\Phi^{\dag},A_{\Delta_1}^*))
  \le \min_{P_Y: \tilde{\calD}_{\Delta_0}(P_Y,P_0) - \tilde{\calD}_{\Delta_1}(P_Y,P_1) \le a}
\tilde{\calD}_{\Delta_1}(P_Y,P_1).
\end{align}
For the lower bound,
\begin{align}
\label{LB_P_FN}
P_{\mbox{\tiny FN}}(\Phi^{\dag},A_{\Delta_1}^*) \ge & (n+1)^{- |\calA \|\calA-1|} \sum_{\by: \tilde{\calD}_{\Delta_0}^n(\hat{P}_{\by},P_0) - \tilde{\calD}_{\Delta_1}^n(\hat{P}_{\by},P_1) < a} e^{-n[\hat{H}_{\by}  +
\tilde{\calD}_{\Delta_1}^n(\hat{P}_{\by},P_1)]}\nonumber\\
 \ge & (n+1)^{- |\calA \|\calA-1|}  \max_{\hat{P}_{\by}: \tilde{\calD}_{\Delta_0}^n(\hat{P}_{\by},P_0) - \tilde{\calD}_{\Delta_1}^n(\hat{P}_{\by},P_1) < a} e^{-n
\tilde{\calD}_{\Delta_1}^n(\hat{P}_{\by},P_1)}\nonumber\\
=  &  (n+1)^{- |\calA \|\calA-1|}  \exp\left\{-n\left(\min_{\hat{P}_{\by}: \tilde{\calD}_{\Delta_0}^n(\hat{P}_{\by},P_0) - \tilde{\calD}_{\Delta_1}^n(\hat{P}_{\by},P_1) < a}
\tilde{\calD}_{\Delta_1}^n(\hat{P}_{\by},P_1)\right)\right\}.
\end{align}
Then
\begin{align}
\label{liminf}
- \liminf_{n\rightarrow \infty} \frac{1}{n} \ln(P_{\mbox{\tiny FN}}(\Phi^{\dag},A_{\Delta_1}^*)) \ge & \lim_{n \rightarrow \infty} \tilde{\calD}_{\Delta_1}^n(\hat{P}_{\by},P_1) \nonumber\\
=  & \min_{P_Y: \tilde{\calD}_{\Delta_0}(P_Y,P_0) - \tilde{\calD}_{\Delta_1}(P_Y,P_1) \le a}
\tilde{\calD}_{\Delta_1}(P_Y,P_1),
\end{align}
where $\hat{P}_{\by}$ is a properly chosen PMF, belonging to the set $\{\tilde{\calD}_{\Delta_0}^n(\hat{P}_{\by},P_0) - \tilde{\calD}_{\Delta_1}^n(\hat{P}_{\by},P_1) < a\}$ for every $n$, and  such that $\hat{P}_{\by} \rightarrow P_Y^*$ where\footnote{By Property \ref{property_convex}, set $\{\tilde{\calD}_{\Delta_0}^n(\hat{P}_{\by},P_0) - \tilde{\calD}_{\Delta_1}^n(\hat{P}_{\by},P_1) < a\}$ is dense in $\{P_Y: \tilde{\calD}_{\Delta_0}(P_Y,P_0) - \tilde{\calD}_{\Delta_1}(P_Y,P_1) \le a\}$ and then such a sequence of PMFs can always be found.}
\begin{equation}
P_Y^* = \arg\min_{P_Y: \tilde{\calD}_{\Delta_0}(P_Y,P_0) - \tilde{\calD}_{\Delta_1}(P_Y,P_1) \le a}
\tilde{\calD}_{\Delta_1}(P_Y,P_1).
\end{equation}
By combining \eqref{limsup} and \eqref{liminf}, we get
%
\begin{align}
\label{epsilon_FN_Bayes}
\varepsilon_{\mbox{\tiny FN}} = - \lim_{n\rightarrow \infty} \frac{1}{n} \ln(P_{\mbox{\tiny FN}}(\Phi^{\dag},A_{\Delta_1}^*)) = \min_{P_Y: \tilde{\calD}_{\Delta_0}(P_Y,P_0) - \tilde{\calD}_{\Delta_1}(P_Y,P_1) \le a}
\tilde{\calD}_{\Delta_1}(P_Y,P_1).
\end{align}
Therefore, from \eqref{UB_P_FN} and \eqref{LB_P_FN}  we have
\begin{align}
\label{epsilon_FN_Bayes}
P_{\mbox{\tiny FN}}(\Phi^{\dag},A_{\Delta_1}^*) \doteq  \exp\left\{-n\left(\min_{\hat{P}_{\by}: \tilde{\calD}_{\Delta_0}^n(\hat{P}_{\by},P_0) - \tilde{\calD}_{\Delta_1}^n(\hat{P}_{\by},P_1) < a}
\tilde{\calD}_{\Delta_1}^n(\hat{P}_{\by},P_1)\right)\right\},
\end{align}
and the limit of $\frac{1}{n} \ln P_{\mbox{\tiny FN}}$ exists and is finite.

Similar bounds can be derived for the FP probability, resulting in
%
\begin{align}
\label{epsilon_FP_Bayes}
P_{\mbox{\tiny FP}}(\Phi^*,A_{\Delta_0}^*)  \doteq \exp\left\{-n\left(\min_{\hat{P}_{\by}: \tilde{\calD}_{\Delta_0}^n(\hat{P}_{\by},P_0) - \tilde{\calD}_{\Delta_1}^n(\hat{P}_{\by},P_1) \ge a}
\tilde{\calD}_{\Delta_0}^n(\hat{P}_{\by},P_0)\right)\right\},
\end{align}
and in particular
%
\begin{align}
\label{epsilon_FP_Bayes}
\varepsilon_{\mbox{\tiny FP}} = - \lim_{n\rightarrow \infty} \frac{1}{n} \ln(P_{\mbox{\tiny FP}}(\Phi^*,A_{\Delta_0}^*)) = \min_{P_Y: \tilde{\calD}_{\Delta_0}(P_Y,P_0) - \tilde{\calD}_{\Delta_1}(P_Y,P_1) \ge a}
\tilde{\calD}_{\Delta_0}(P_Y,P_0).
\end{align}
From  \eqref{epsilon_FP_Bayes}, we see that, as argued,
the profile $(\Phi^{\dag},(A_{\Delta_0}^*, A_{\Delta_1}^*))$ leads to a FP exponent always
at least as large as $a$.

We are now ready to evaluate the asymptotic behavior of the payoff of the Bayesian detection  game:
%
%
\begin{align}
u = & P_{\mbox{\tiny FN}}(\Phi^{\dag},A_{\Delta_1}^*) + e^{a n} P_{\mbox{\tiny FP}}(\Phi^{\dag},A_{\Delta_0}^*)\nonumber\\
\doteq & \max \{P_{\mbox{\tiny FN}}(\Phi^{\dag},A_{\Delta_1}^*), e^{a n} P_{\mbox{\tiny FP}}(\Phi^{\dag},A_{\Delta_0}^*)\}\nonumber\\
%
%
\doteq & \exp\left\{-n \min \left(\min_{\hat{P}_{\by}: \tilde{\calD}_{\Delta_0}^n(\hat{P}_{\by},P_0) - \tilde{\calD}_{\Delta_1}^n(\hat{P}_{\by},P_1) < a}
\tilde{\calD}_{\Delta_1}^n(\hat{P}_{\by},P_1), \min_{\hat{P}_{\by}: \tilde{\calD}_{\Delta_0}^n(\hat{P}_{\by},P_0) - \tilde{\calD}_{\Delta_1}^n(\hat{P}_{\by},P_1) \ge a}
(\tilde{\calD}_{\Delta_0}^n(\hat{P}_{\by},P_0) - a)\right)\right\}\nonumber\\
= & \exp\left\{-n \min_{P_{\by}} \left(\max\left\{
\tilde{\calD}_{\Delta_1}^n(\hat{P}_{\by},P_1), (\tilde{\calD}_{\Delta_0}^n(\hat{P}_{\by},P_0) - a)\right\}\right)\right\}\nonumber\\
\doteq & \exp\left\{-n \min_{P_Y} \left(\max\left\{
\tilde{\calD}_{\Delta_1}(P_Y,P_1), (\tilde{\calD}_{\Delta_0}(P_Y,P_0) - a)\right\}\right)\right\},
%
\label{infimum_u}
\end{align}
where the asymptotic equality in the last line follows from the density of the set of empirical probability distributions of $n$-length sequences in the probability simplex and from the continuity of the to-be-minimized expression in round brackets as a function of $P_Y$.
%

\section{Source distinguishability}
\label{App-SourceDist}

This appendix contains the proofs for Section \ref{sec.SA_limitingPerf}.

\subsection{Proof of Theorem \ref{theorem_EMD_FA_B}}
\label{Appx.theorem_EMD_FA_B}

The theorem directly follows from Theorem \ref{theo_e.e_B}. In fact, by letting
\begin{equation}
e_a(P_Y) =  \max\left\{
\tilde{\calD}_{\Delta_1}(P_Y,P_1), \tilde{\calD}_{\Delta_0}(P_Y,P_0) - a\right\},
\end{equation}
$a \ge 0$, the limit in \eqref{best_e_e_B} can be derived as follows:
\begin{align}
  \lim_{a \rightarrow 0} \hspace{0.1cm} \min_{P_Y} e_a(P_Y) & =  \min_{P_Y} \hspace{0.1cm} \lim_{a \rightarrow 0} e_a(P_Y) \nonumber\\
 & =  \min_{P_Y} \left(\max\left\{
\tilde{\calD}_{\Delta_1}(P_Y,P_1), \tilde{\calD}_{\Delta_0}(P_Y,P_0)\right\}\right),
\end{align}
where the order of limit and minimum can be exchanged because of the uniform convergence of $e_a(P_Y)$ to $e_{0}(P_Y)$ as $a$ tends to $0$.

\subsection{Proof of Corollary \ref{cor_Gamma_d_metric}}
\label{Appx.cor_Gamma_d_metric}

The corollary can be proven by exploiting the fact that, when $d$ is a
metric, the {\em EMD} is a metric and then
$\text{\em EMD}_d(P_0, P)$ satisfies the triangular inequality. In this case, it is easy to argue that  the $P_Y$ achieving the minimum in \eqref{A_winning_region} is the one for which the triangular relation holds at the equality, which corresponds to the convex combination of $P_0$ and $P$ (i.e., the PMF lying on the straight line between $P_0$ and $P$) with combination coefficient $\alpha$ such that $\text{\em EMD}_d(P_0, P_Y)$ (or equivalently, by symmetry, $\text{\em EMD}_d(P_Y, P_0)$) is exactly equal to $\Delta_0$.

Formally,  let $X \sim P_0$ and $Z \sim P$.
We want to find the PMF $P_Y$ which solves
\begin{equation}
\label{ref_rewriting}
\min_{P_Y : \text{\em EMD}_{d}(P_Y,P_0) \le \Delta_0} \text{\em EMD}_{d}(P_Y,P).
\end{equation}
For any $Y \sim P_Y$ and any choice of $P_{XY}$ and $P_{YZ}$ (that is,  $P_{Y|X}$ and $P_{Z|Y}$), by exploiting the triangular inequality property of the distance,  we can write
\begin{align}
E_{XZ} d(X,Z) \le E_{XY} d(X,Y) + E_{YZ} d(Y,Z),
\end{align}
where $P_{XZ}$ can be any joint distribution with marginals $P_0$ and $P$. Then,
\begin{equation}
\text{\em EMD}(P_0, P) \le E_{XY} d(X,Y) + E_{YZ} d(Y,Z).
\end{equation}
From the arbitrariness of the choice of $P_{XY}$ and $P_{YZ}$, if we let
 $P_{XY}^*$ and $P_{YZ}^*$ be the joint distributions achieving the {\em EMD} between $X$ and $Y$, and $Y$ and $Z$, we get
%
\begin{equation}
\label{relation_between_EMD}
\text{\em EMD}(P_0, P) \le \text{\em EMD}(P_0, P_Y)  + \text{\em EMD}(P_Y, P).
\end{equation}
From the above relation, we can derive the following lower bound for the to-be-minimized quantity in  \eqref{ref_rewriting}:
\begin{align}
\label{lower_bound}
\text{\em EMD}(P_Y, P) \ge & \text{\em EMD}(P_0, P) - \text{\em EMD}(P_0, P_Y)\\
\ge & \text{\em EMD}(P_0, P) - \Delta_0. \label{lower_bound2}
\end{align}
We now show that $P_Y$ defined as in \eqref{optimumPY} achieves the above lower bound while satisfying the constraint $\text{\em EMD}(P_0, P_Y) \le \Delta_0$, and then gets the minimum value in \eqref{ref_rewriting}.

Let $P_{XZ}^*$ be the joint distribution achieving the \text{\em EMD} between $X$ and $Z$. Then,  $E_{XZ}^* d(X,Z) = \text{\em EMD}(P_0, P)$ (where the star on the apex indicates that the expectation is taken under $P_{XZ}^*$).
Given the marginal $P_Y = \alpha P_0 + (1-\alpha)P$, we can define $P_{XY}$ and $P_{YZ}$, starting from $P_{XZ}^*$, as in the proof of  Theorem \ref{theo_e_fn_d} (\eqref{clever_joint} and \eqref{clever_joint2}). With this choice,
$E_{XY} d(X,Y) = (1-\alpha) \text{\em EMD}(P_0, P)$ and $E_{YZ} d(Y,Z) = \alpha \text{\em EMD}(P_0, P)$. Then, for the value of $\alpha$ in \eqref{optimumPY} we have that
$E_{XY} d(X,Y) = \Delta_0$ and 
\begin{equation}
E_{YZ} d(Y,Z) =  \text{\em EMD}(P_0, P) - \Delta_0. \label{key_relation2}
\end{equation}
%
By combining \eqref{key_relation2} and \eqref{lower_bound2}, we argue that $\text{\em EMD}(P_Y, P) = \text{\em EMD}(P_0, P) - \Delta_0$\footnote{We also argue that the choice made for $P_{YZ}$ minimizes the expected distortion between $Y$ and $Z$, i.e., it yields $E_{YZ} d(Y,Z) = \text{\em EMD}(P_Y, P)$. Furthermore, being $E_{XY} d(X,Y) = \Delta_0$, it holds $\text{\em EMD}(P_Y, P) = \text{\em EMD}(P_0, P) - E_{XY} d(X,Y)$ and then, from the triangular inequality in \eqref{relation_between_EMD}, it follows that  $\text{\em EMD}(P_0, P_Y) = E_{XY} d(X,Y)= \Delta_0$.}.
Therefore, $P_Y$ in \eqref{optimumPY} solves  \eqref{ref_rewriting}.
%

To prove the second part of the corollary, we just need to observe that a PMF $P$ belongs to the indistinguishability set in \eqref{A_winning_region} if and only if
\begin{equation}
\text{\em EMD}(P_Y, P) =  \text{\em EMD}(P_0, P) - \Delta_0 \le \Delta_1,
\end{equation}
that is $\text{\em EMD}(P_0, P) \le \Delta_0 + \Delta_1$.\\

From the above proof, we notice that,  for any $P$ in the set in \eqref{A_winning_region_d}, i.e., such that ${\text{\em EMD}}_d(P_0,P) \le \Delta_0 + \Delta_1$, the PMF $P_Y= \alpha P_0 + (1-\alpha)P$ with $\alpha$ as in \eqref{optimumPY} satisfies  {\em EMD}$(P_Y,P_0) = \Delta_0$ and {\em EMD}$(P_Y,P_1) = \Delta_1$ for any choice of $d$. Then, when $d$ is not a metric, the region  in \eqref{A_winning_region_d} is contained in the indistinguishability region.

\subsection{Proof of Corollary \ref{cor_L_p_p}}
\label{Appx.cor_L_p_p}

By inspecting the minimization in \eqref{A_winning_region}, we see that for any source $P$ that cannot be distinguished from $P_0$, it is possible to find a source $P_Y$ such that $\text{\em EMD}_{d}(P_Y,P) \le \Delta_1$ and $\text{\em EMD}_{d}(P_Y,P_0) \le \Delta_0$. In order to prove the corollary, we need to show that such $P$ lies inside the set defined in \eqref{A_winning_region_L_2}.
%

We give the following definition.
Given two random variables $X$ and $Y$, the H{\"o}lder inequality applied to the expectation function (\cite{keyInequalities}) reads:
\begin{equation}
E_{XY} |XY| \le \big(E_X [|X|^r]\big)^{1/r} \big(E_Y [|Y|^q]\big)^{1/q},
\end{equation}
where $r \ge 1$ and $q = r/(r-1)$, namely, the H{\"o}lder conjugate of $r$.

We use the notation $E_{XY}^*$  for the expectation of the pair $(X,Y)$ when the probability map is the one achieving the $\text{\em EMD}(P_X, P_Y)$, namely $P_{XZ}^*$.
Then, we can write:
\begin{align}
\text{\em EMD}_{L_p^p}(P_0, P) & =  E_{XZ}^* [||X - Z||^p]\nonumber\\
& {\stackrel{(a)}{\le}} E_{XYZ}^* [(||X - Y|| + ||Y - Z||)^p]\nonumber\\
& {\stackrel{(b)}{\le}} E_{XYZ}\left[||X - Y||^p + ||Y - Z||^p + p \cdot ||X - Y||^{p-1} \hspace{0.1cm}||Y - Z|| +\right. \nonumber\\
& \left. \hspace{1.5cm} + p(p -1)/2 \cdot ||X - Y||^{p-2} \hspace{0.1cm}||Y - Z||^2 + ..... + p \cdot||X - Y|| \hspace{0.1cm}||Y - Z||^{p-1}\right]\nonumber\\
& = E_{XYZ}[||X - Y||^p] + E_{XYZ}[||Y - Z||^p] + p \cdot E_{XYZ}[||X - Y||^{p-1} \hspace{0.1cm}||Y - Z||] + \nonumber\\
&  \hspace{1cm} + p(p -1)/2 \cdot E_{XYZ}[||X - Y||^{p-2} \hspace{0.1cm}||Y - Z||^2] + ..... + p \cdot E_{XYZ}[||X - Y|| \hspace{0.1cm}||Y - Z||^{p-1}]\nonumber\\
& {\stackrel{(c)}{\le}} E_{XYZ}[||X - Y||^p] + E_{XYZ}[||Y - Z||^p] + p \cdot {(E_{XYZ}[||X - Y||^p])}^{\frac{p-1}{p}} (E_{XYZ}[||Y - Z||^p])^{\frac{1}{p}}\nonumber\\
&  \hspace{1cm} + p(p -1)/2 \cdot (E_{XYZ}[||X - Y||^{p})^{\frac{p-2}{p}} (E_{XYZ}[||Y - Z||^p])^{\frac{2}{p}} + ...\nonumber\\
& \hspace{3cm}... + p \cdot (E_{XYZ}[||X - Y||^p])^{\frac{1}{p}} (E_{XYZ}[||Y - Z||^{p}])^{\frac{p-1}{p}}\nonumber\\
& =  E_{XY}[||X - Y||^p] + E_{YZ}[||Y - Z||^p] + p \cdot {(E_{XY}[||X - Y||^p])}^{\frac{p-1}{p}} (E_{YZ}[||Y - Z||^p])^{\frac{1}{p}}\nonumber\\
&  \hspace{1cm} + p(p -1)/2 \cdot (E_{XY}[||X - Y||^{p}])^{\frac{p-2}{p}} (E_{YZ}[||Y - Z||^p])^{\frac{2}{p}} + ...\nonumber\\
& \hspace{3cm}... + p \cdot (E_{XY}[||X - Y||^p])^{\frac{1}{p}} (E_{YZ}[||Y - Z||^{p}])^{\frac{p-1}{p}}\nonumber\\
& = \left((E_{XYZ}[||X - Y||^p])^{1/p} + (E_{XYZ}[||Y - Z||^p])^{1/p}\right)^{p}\nonumber\\
& \le \left(\Delta_0^{1/p} + \Delta_1^{1/p}\right)^p,
\end{align}
where in $(a)$ we considered the joint distribution $P_{XYZ}$ such that
$\sum_{Z} P_{XYZ} = P_{XY}^*$, $\sum_{X} P_{XYZ} = P_{YZ}^*$ (and,
consequently, $\sum_{Y} P_{XYZ} = P_{XZ}^*$) and in $(b)$ we developed the
$p$-power of the binomial (binomial theorem). Finally, in $(c)$, we applied
the H\"older's inequality to the various terms of
Newton's binomial: specifically, for each term $E_{XYZ}[||X - Y||^{p-t} \hspace{0.1cm}||Y - Z||^t]$, with $t = 1,..,p-1$, the H\"older inequality is applied with $r = p/(p-t)$ (and $q = r/(r-1)$).

\bibliographystyle{IEEEtran}
\bibliography{A-D-game}

\end{document}